\documentclass[draftclsnofoot,12pt,onecolumn]{IEEEtran}
\usepackage{bbm}
\usepackage{amssymb,amsmath,amsthm}
\usepackage{graphicx,epstopdf,psfrag,url,cite,xcolor,multirow,float,array}
\usepackage{caption-patch}
\usepackage[small,bf]{caption}
\usepackage[font=footnotesize]{subfig}
\usepackage{diagbox}

\usepackage{algorithm}
\usepackage{algpseudocode}

\newtheoremstyle{mythm}
{\topsep}   
{\topsep}   
{\itshape}      
{0pt}       
{\bfseries} 
{:}         
{5pt plus 1pt minus 1pt}    
{\thmname{#1}\thmnumber{ #2}\thmnote{ (#3)}}
\theoremstyle{mythm}

\newtheorem{proposition}{Proposition}

\newcommand{\singlesize}{0.6} 

\newcommand{\bm}[1]{\boldsymbol{#1}}

\IEEEoverridecommandlockouts

\begin{document}

\title{Optimization-driven Hierarchical Learning Framework for Wireless Powered Backscatter-aided Relay Communications}
\author{Shimin Gong, Yuze Zou, Jing Xu, Dinh Thai Hoang, Bin Lyu, and Dusit Niyato
\thanks{Shimin Gong is with School of Intelligent Systems Engineering, Sun Yat-sen University, Guangzhou 510275, China (e-mail: gongshm5@mail.sysu.edu.cn). Yuze Zou and Jing Xu are with the School of Electronic Information and Communications, Huazhong University of Science and Technology, Wuhan 430074, China (email: \{zouyuze,xujing\}@hust.edu.cn). Dinh Thai Hoang is with the Faculty of Engineering and Information Technology, University of Technology Sydney, Ultimo NSW 2007, Australia (email: hoang.dinh@uts.edu.au). Bin Lyu is with Key Laboratory of Ministry of Education in Broadband Wireless Communication and Sensor Network Technology, Nanjing University of Posts and Telecommunications, Nanjing 210003, China. Dusit Niyato is with School of Computer Science and Engineering, Nanyang Technological University, Singapore (email: dniyato@ntu.edu.sg).}
}

\maketitle
\thispagestyle{empty}

\begin{abstract}
In this paper, we employ multiple wireless-powered relays to assist information transmission from a multi-antenna access point to a single-antenna receiver. The wireless relays can operate in either the passive mode via backscatter communications or the active mode via RF communications, depending on their channel conditions and energy states. We aim to maximize the overall throughput by jointly optimizing the access point's beamforming and the relays' radio modes and operating parameters. Due to the non-convex and combinatorial structure, we develop a novel optimization-driven hierarchical deep deterministic policy gradient (H-DDPG) approach to adapt the beamforming and relay strategies dynamically. The optimization-driven H-DDPG algorithm firstly decomposes the binary relay mode selection into the outer-loop deep $Q$-network (DQN) algorithm and then optimizes the continuous beamforming and relaying parameters by using the inner-loop DDPG algorithm. Secondly, to improve the learning efficiency, we integrate the model-based optimization into the DDPG framework by providing a better-informed target estimation for DNN training. Simulation results reveal that these two special designs ensure a more stable learning and achieve a higher reward performance, up to nearly 20\%, compared to the conventional DDPG approach.
\end{abstract}

\section{Introduction}
Backscatter communications technology has been recently proposed as a solution to improve the energy- and spectrum-efficiency of wireless networks~\cite{lxmag18,symbiotic}. The backscatter radio works in the passive mode by adapting the antenna's load impedance to reflect the incident or ambient RF signals. It is featured with extremely low power consumption and thus sustainable by the wireless power transfer~\cite{yangiot19,zhou19}. The low power consumption also comes with a price. Due to limited signal processing capability, the backscatter radio typically has a low data rate and are vulnerable to channel variations. Hence, it is more preferable to take advantages of both RF and backscatter communications in a hybrid radio network, where hybrid radios can switch between the active and passive modes according to their energy and channel conditions, e.g.,~\cite{hoang17cr} and~\cite{tccn19}.

In a hybrid radio network, the backscatter radios can be used as passive relays to assist RF communications~\cite{ieeenetwork}. Typically, there are two passive relaying strategies. The first case is similar to the conventional decode-and-forward (DF) relay scheme. The hybrid relay decodes the information in the first hop and then forwards it to the receiver in the second hop by backscattering or RF communications. For example, the authors in~\cite{lyb-arelay} and~\cite{lyb-arelay2} employed a gateway device to decode the backscattered information from the user device and then forward the information to the access point by RF communications. An opportunistic DF relay scheme is proposed in~\cite{twobds} where the relay node not only decodes and forwards the source information but also opportunistically transmits its own information by backscatter communications. The authors in~\cite{R5} studied the cooperation between an active radio and a backscatter radio to deliver information to the access point. The challenge in this case lies in that the relay's signal decoding requires high power consumption. This may hinder the relay to join relaying communications.

The other case is similar to the conventional amplify-and-forward (AF) scheme, where the passive relay instantly reflects the incident RF signals. As such, the RF channel between the transceivers can be enhanced in favor of RF communications by optimizing the relays' reflection coefficients, e.g.,~\cite{symbiotic,yang18,tccn19}. Moreover, the relay does not need to decode information and hence has very little power consumption. This will encourage cooperative relay communications. Focusing on such an AF-alike passive relay scheme, the bit-error-rate (BER) performance was analyzed in~\cite{zhouber19} and~\cite{zhou20} and compared with the DF-alike relay scheme, revealing that both schemes perform similarly under ideal conditions while the DF-alike scheme becomes worse off in more practical conditions. The authors in~\cite{luxiao19} proposed to optimize the mode switching between the AF-alike passive relay scheme and the conventional DF scheme to improve the transmission success probability and ergodic capacity. To maximize the transmission performance in a hybrid radio network, the above-mentioned works generally formulate a joint optimization problem, typically involving the radio's mode selection, the choice of complex reflection coefficient, the energy harvesting parameters, transmission scheduling, power allocation, energy and information beamforming strategies.

In this paper, we focus on the energy- and spectrum-efficient passive relay communications in the second case. Different from prior studies in~\cite{yang18,zhouber19,zhou20,luxiao19}, we consider a more challenging scenario where multiple hybrid radios jointly assist the RF communications from a hybrid access point (HAP) to an active receiver. Each radio can choose between the active and passive mode independently. Hence, we expect that the RF communications between the active transceivers will be assisted by both the active and passive relays simultaneously~\cite{gc19xie}. We aim to improve the throughput performance by exploiting both the relays' radio diversity gain and multi-user cooperation gain. A similar hybrid relay network has been studied in our previous work in~\cite{iot20}, where the relays' mode selection is approximately optimized by a set of heuristic algorithms to improve the overall relay performance. The difficulty of the optimization problem firstly lies in the combinatorial structure due to the relays' mode selection. Even with the fixed relay mode, the joint optimization of beamforming and relaying strategies is still challenged due to complicated couplings among multiple relays. Compared with optimization methods, the machine learning (ML) approaches have shown better flexibility and robustness to address complex problems with imprecise modeling, uncertain dynamics, and high-dimensional variables. For example, the authors in~\cite{dsa-rl} proposed a Markov Decision Process (MDP) and applied an online reinforcement learning (RL) method to learn the radio's spectrum access decision in an ambient backscatter system, considering the dynamics in the channel condition, energy storage, and traffic demand. Deep reinforcement learning (DRL) has been applied in~\cite{jam-hoang} and~\cite{rl-jam} to defend the hybrid radio against a smart jammer. ML methods have also been used for performance maximization in hybrid radio networks, e.g.,~by optimizing transmission scheduling~\cite{drl-time}, power allocation~\cite{q-power}, and radio mode selection~\cite{iccc-ddpg}. However, the above-mentioned applications of ML methods are still unsatisfactory in practice, mainly due to the demand for a large data set for offline training and a slow convergence speed in online learning.

In this paper, we propose a novel DRL approach for throughput maximization by adapting the relays' mode selection, the beamforming strategy, and time allocation for wireless power transfer. A close inspection into two typical DRL algorithms, e.g.,~deep $Q$-learning (DQN) and deep deterministic policy gradient (DDPG)~\cite{ddpg}, reveals that they both rely on a double $Q$-network structure~\cite{luong18}. The online $Q$-network provides the value estimation for each state-action pair using a set of deep neural network (DNN), while the target $Q$-network generates a target value for the online $Q$-network to follow and adapt. The DNN parameters of the target $Q$-network are regularly copied from the online $Q$-network, which brings strong coupling between two $Q$-networks and results in performance fluctuations. Besides, both $Q$-networks are typically randomly initialized, which requires a long warm-up period to stabilize the learning. These observations motivate us to design a novel hierarchical DDPG (denoted as H-DDPG) algorithm to improve the learning efficiency. One novel design of the H-DDPG is to reduce the action space by proposing a hierarchical structure, i.e., the binary mode selection is optimized by the outer-loop DQN algorithm while the continuous variables are left for inner-loop DDPG algorithm. The second design is to provide the inner-loop DDPG a better-informed target for DNN training, by integrating the model-based optimization into the model-free DDPG algorithm.

A preliminary work on the hierarchical learning framework has been presented in our conference paper in~\cite{hddpg}, where we consider the power-splitting (PS) protocol for the energy harvesting relays. In this work, we focus on a different time-switching (TS) protocol and analyze more detailed designs for the optimization-driven H-DDPG algorithm. We also try to integrate different performance lower bounds into the DDPG framework, and verify its robustness and learning efficiency. To be specific, our main contributions in this paper are summarized as follows:
\begin{enumerate}
  \item A multi-relay-assisted communication model is proposed for a hybrid radio network, where each relay can optimize its radio mode to assist RF communications from the HAP to the receiver. The passive relays enhance the RF channels, while the active relays amplify and forward the information using the energy harvested from the HAP in the TS protocol.
  \item A throughput maximization problem is formulated to optimize the time allocation and beamforming strategies, as well as the relaying strategy. Though it is difficult to solve optimally, we propose a tractable approximation as the lower bound of the original optimization problem, which can be optimally solved by the monotonic optimization algorithm.
  \item A hierarchical learning framework is proposed to solve the throughput maximization problem. We update the relays' mode selection in the outer-loop DQN algorithm, while use the inner-loop DDPG algorithm to adapt the relays' reflection coefficients, the HAP's beamforming and time allocation strategies.
  \item We further design the optimization-driven H-DDPG algorithm to improve the learning efficiency, which employs a model-based optimization module to find a lower bound on the target $Q$-value by solving an approximation of the original problem. Extensive numerical results reveal that the optimization-driven H-DDPG algorithm achieves significantly higher reward and more stable learning performance compared to the model-free DDPG algorithm.
\end{enumerate}

The rest of this paper is organized as follows. Section~\ref{sec_model} presents the multi-relay-assisted communication model. Section~\ref{sec_lbds} proposes a throughput maximization problem and derives its lower bound. In Section~\ref{sec_H-DDPG}, we integrate the lower bound into the DRL framework and present the optimization-driven H-DDPG algorithm. Numerical evaluations and conclusions are presented in Sections~\ref{sec_sim} and~\ref{sec_con}, respectively.

\section{System model}\label{sec_model}

We allow a group of single-antenna user devices, denoted by the set $\mathcal{N}=\{1,2,\ldots,N\}$, to harvest energy from a multi-antenna HAP in the TS protocol and then assist the information transmission from the HAP to its receiver. The HAP has constant power supply and fixed transmit power $p_t$. Its transmit beamforming vector can be tuned to optimize the wireless power transfer to different relays. The relay-assisted information transmission follows a two-hop half-duplex protocol. Each RF-powered relay has a dual-mode radio structure that can switch between the passive backscatter communications and the active RF communications, e.g.,~\cite{ieeenetwork} and~\cite{iot20}. An example with two relays (one in passive mode and the other in active mode) is shown in Fig.~\ref{fig:hybridmodel}. Let ${\bf f}_0 \in \mathbb{C}^K$ and ${\bf f}_n \in\mathbb{C}^K$ denote the complex channels from the HAP (with $K$ antennas) to the receiver and to the $n$-th relay, respectively. The complex channels from relay-$n$ to relay-$m$ and to the receiver are given by $z_{nm}$ and $g_n$, respectively.


\subsection{Time-switching (TS) Protocol for Active Relays}
As shown in Fig.~\ref{fig:hybridmodel}, the TS protocol assigns a dedicated sub-slot with length $w$ for the HAP to beamform RF power to the active relays. The other part $(1-w)$ of a time slot is further divided into two equal sub-slots with length $t=(1-w)/2$ for the active relays to receive and collaboratively forward signals, respectively in two hops. In the first hop, the HAP sets its signal beamformer ${\bf w}_1$ and transmits information to the active relays and the receiver directly. In the second hop, the active relays collaboratively beamform the received signals to the receiver. A strong direct link ${\bf f}_0$ between the HAP and the receiver may exist in both hops and contribute significantly to the overall throughput, e.g.,~\cite{iot20} and~\cite{gsmwcnc19}. Hence, the HAP can transmit the same information $s$ in two hops to enhance the reliability of signal reception by maximal ratio combining (MRC) at the receiver. Let $({\bf w}_1,{\bf w}_2)$ denote the HAP's signal beamforming vectors in two hops. It is clear that ${\bf w}_1$ and ${\bf w}_2$ are not necessarily the same in two hops.

\subsection{Channel Enhancement via Passive Relays}
Let binary variable $b_k\in\{0,1\}$ denote the radio mode of relay-$k$, i.e.,~$b_k=0$ and $b_k=1$ indicate the active and passive relays, respectively. Correspondingly, we can define $\mathcal{N}_a$ and $\mathcal{N}_b$ as the sets of all active and passive relays, respectively. Let $\hat{\bf f}_0$ and $\hat{\bf f}_n$ denote the backscatter-assisted channels from the HAP to the receiver and to the active relay-$n$, respectively. Following a similar model in~\cite{iot20}, the enhanced channels $\hat{\bf f}_0$ and $\hat{\bf f}_n$ can be represented as follows:
\begin{align}
&\hat{\bf f}_0 = {\bf f}_0 + \sum_{k\in\mathcal{N} } b_k g_k \Gamma_k {\bf f}_k ,~\label{equ_channel_direct}\\
&\hat{\bf f}_n = {\bf f}_n + \sum_{k\in\mathcal{N}} b_k z_{kn} \Gamma_k {\bf f}_k, \quad \forall n \in \mathcal{N}.\label{equ_channel_relay}
\end{align}
It is clear that $\hat{\bf f}_0$ and $\hat{\bf f}_n$ depend on not only the relays' mode selection but also the complex reflection coefficient $\Gamma_k$ of each passive relay. Given a fixed set of passive relays and their reflection coefficients, we can evaluate the direct channel $\hat{\bf f}_0$ and all relaying channels $\hat{\bf f}_n$ of the active relays, and then we can focus on the relay optimization with only active relays.

\begin{figure}[t]
\centering
\includegraphics[width=0.95\textwidth]{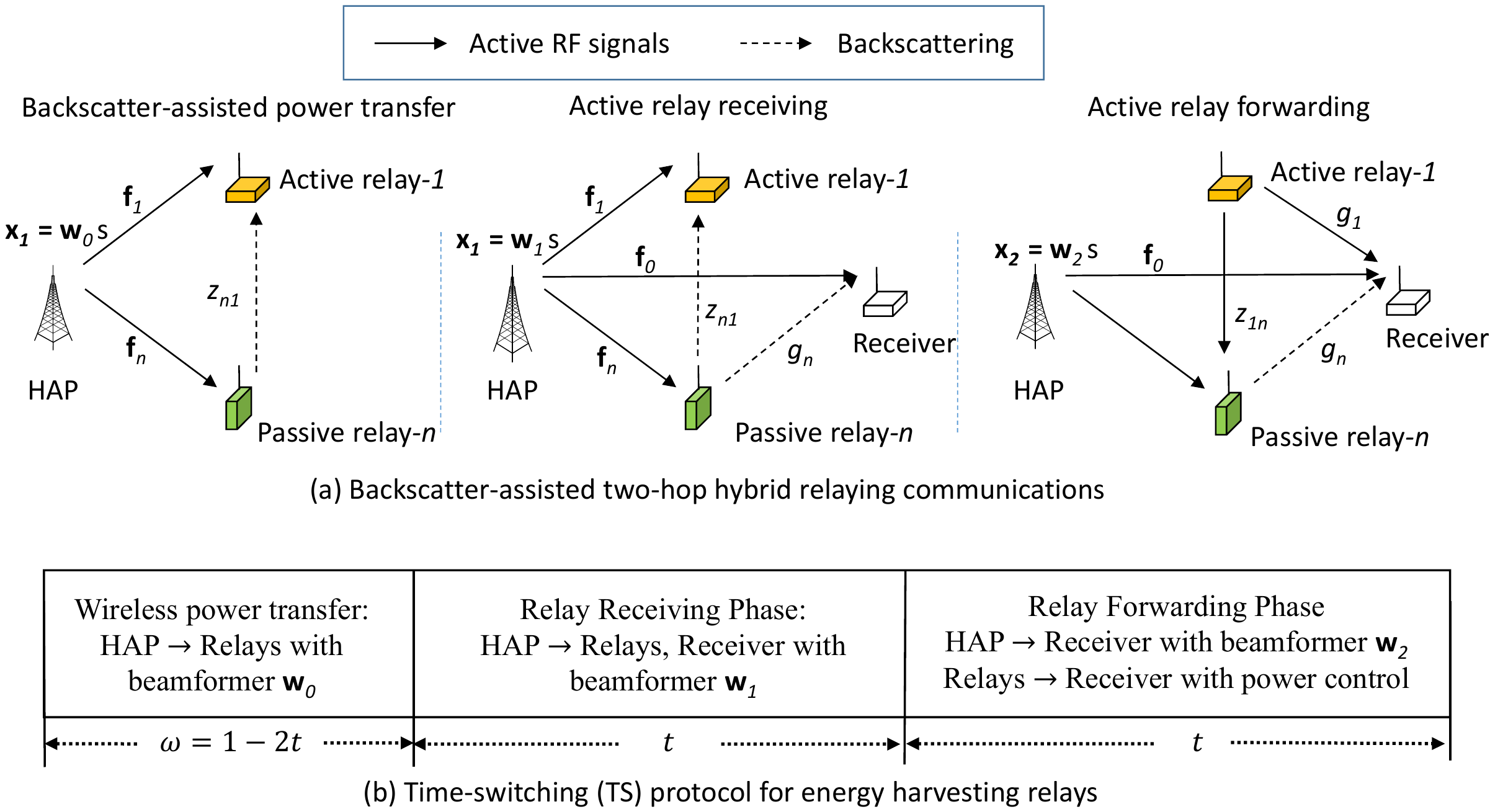}
\caption{Backscatter-assisted two-hop hybrid relaying communications.}\label{fig:hybridmodel}
\end{figure}



\subsection{Problem Formulation}

In the first hop, the HAP's beamforming information ${\bf x}_1= \sqrt{p_t}{\bf w}_1 s$ can be received by both the active relays and the target receiver directly. The SNR at the receiver can be denoted by $\gamma_1 =  p_t|\hat{\bf f}_0^H {\bf w}_1|^2,$ where $\hat{\bf f}^H_0$ is the Hermitian transpose of the enhanced channel $\hat {\bf f}_0$ in~\eqref{equ_channel_direct} and the noise power is normalized to one. In the second hop, all active relays jointly amplify and forward the information to the receiver. Each active relay-$n$ sets a different power amplifier coefficient $x_n\in(0,1)$. Meanwhile, the HAP beamforms the same information symbol ${\bf x}_2 =\sqrt{p_t}{\bf w}_2 s $ to the receiver. Note that ${\bf w}_2$ can be simply aligned to the channel $\hat{\bf f}_0$. Hence, the received signal $r_d$ at the receiver is given by $r_d  = \sum_{n=1}^{N} \hat g_n  x_n r_n + \sqrt{p_t}||\hat{\bf f}_0|| s + v_d$, where $\hat g_n$ denotes the enhanced channel from the active relay-$n$ to the receiver. Similar to~\cite{iot20}, the SNR in the second hop can be simplified as follows:
\begin{equation}
\gamma_2 = \frac{\left| \sum_{n\in\mathcal{N}} x_n y_n \hat g_n + \sqrt{p_t}\lVert\hat{\bf f}_0^H\rVert \right|^2}{1 + \sum_{n\in\mathcal{N}}|x_n \hat g_n|^2 },
\label{equ:gamma}
\end{equation}
where $y_n \triangleq \sqrt{p_t} \hat{\bf f}_n^H {\bf w}_1$ for notational convenience. The power amplifier coefficient of the active relay-$n$ is given by $x_n = \left(\frac{p_n}{1+|y_n|^2}\right)^{1/2}$, where $p_n$ denotes its transmit power. To maximize the overall throughput, we aim to optimize the HAP's beamforming $({\bf w}_0,{\bf w}_1)$, time allocation $t$, and the relays' operating parameters $(b_n,\theta_n)$ in the following throughput maximization problem:
\begin{subequations}\label{prob_bin}
\begin{align}
\max_{ {\bf w}_0, {\bf w}_1, b_n, t, \theta_n} ~&~ t\log_2{(1+\gamma_1+\gamma_2)} \label{obj_sumsnr}\\
s.t.
~&~ ||{\bf w}_0 || \leq 1 \text{ and } ||{\bf w}_1|| \leq 1,\label{con_bvector}\\
~&~ p_n\leq\eta(1/t-2)p_t|\hat{\bf f}_n^H{\bf w}_0|^2,   \forall \,\, n\in\mathcal{N}_a,\label{con_powerbd}\\
~&~ t \in(0,1/2), \label{con_rho}\\
~&~ b_n \in \{0,1\}, \quad \forall\,\, n\in\mathcal{N},\\
~&~ \theta_n \in [0, 2\pi] ,\quad \forall\,\, n\in\mathcal{N}_b.\label{con_gamma}
\end{align}
\end{subequations}
The constraints in~\eqref{con_bvector} denote the HAP's feasible beamforming vectors in two hops. The constraints in~\eqref{con_powerbd} and~\eqref{con_rho} determine the active relays' transmit power in the second hop, which is upper bounded by the energy harvested from the HAP's beamforming in the first phase. The constant $\eta$ denotes the energy harvesting efficiency. The constraint in~\eqref{con_gamma} ensures that the complex reflection coefficient $\Gamma_n=|\Gamma_n|e^{j\theta_n}$ of each passive relay is fully controllable via load modulation~\cite{ieeenetwork,iot20}. From~\eqref{equ_channel_direct} and~\eqref{equ_channel_relay}, we observe that the phase $\theta_n$ is a critical design variable for channel enhancement while $|\Gamma_n|$ can be simply set to its maximum $\Gamma_{\max}$ to enhance the reflected signal power.

\section{Performance Lower Bound}\label{sec_lbds}
The optimization of $(\mathcal{N}_a, \mathcal{N}_b)$ is combinatorial. Even with a fixed radio mode selection $b_n$, the joint optimization of $({\bf w}_0, {\bf w}_1)$ and the passive relays' operating parameters $\theta_n$ is still challenging due to couplings among different active relays. In the following part, with fixed radio mode, we firstly derive a lower bound on~\eqref{prob_bin} by optimizing the beamforming and time allocation strategies. Then this lower bound will be used to devise the optimization-driven learning framework for solving the original optimization problem~\eqref{prob_bin} with enhance efficiency and reward performance.

\subsection{Lower Bound via Monotonic Optimization}

Given the set $\mathcal{N}_b$ of passive relays and the reflection coefficients $\Gamma_{n}$, the enhanced channels for active RF communications can be updated in~\eqref{equ_channel_direct} and~\eqref{equ_channel_relay}. Then, we can formulate the throughput maximization problem with the active relays only:
\begin{subequations}\label{prob_ts}
\begin{align}
\max_{t,{\bf w}_0,{\bf w}_1}~&~t\log_2{(1+\gamma_1+\gamma_2)} \label{obj_ts}\\
s.t.~&~p_n\leq\eta(1/t-2)p_t|\hat{\bf f}_n^H{\bf w}_0|^2, \quad \forall\,\, n\in\mathcal{N}_a,\label{con_ts_power}\\
~&~ 0<t<1/2, ||{\bf w}_0 || \leq 1 \text{ and } ||{\bf w}_1|| \leq 1.
\end{align}
\end{subequations}
Note that an optimal solution to problem~\eqref{prob_ts} is still unavailable due to the non-convex structure. A lower bound on~\eqref{prob_ts} can be useful to evaluate the transmission performance, which leads to the following optimization problem:
\begin{proposition}\label{pro_ts}
A feasible lower bound on~\eqref{prob_ts} is given by the following problem:
\begin{subequations}\label{prob_trans_ts}
\begin{align}
\max_{t,{\bf w}_0,{\bf w}_1}~&~t\log_2{\left(1+p_t||\hat {\bf f}_0||^2 + p_t\sum_{n\in\mathcal{N}_a\cup\{0\}} s_{n,1}\right)}\label{obj_trans_ts}\\
s.t.~&~ \psi_n(1/t-2) s_{n,0} \geq s_{n,1}(1+p_t s_{n,1}), \forall\,\, n\in\mathcal{N}_a,\label{cons_trans_ts}\\
	~&~ 0<t<1/2, ||{\bf w}_0|| \leq 1 \text{ and } ||{\bf w}_1|| \leq 1.\label{con2_trans_ts}
\end{align}
\end{subequations}
where $\psi_n \triangleq \eta p_t  |\hat g_n|^2 ||\hat {\bf f}_0||^2$ and $s_{n,i}\triangleq|\hat{\bf f}_n^H{\bf w}_i|^2$ for $i\in\{0,1\}$ and $n\in{0}\cup \mathcal{N}_a$.
\end{proposition}
\begin{proof}
The proof follows a similar idea as that in~\cite{gsmwcnc19} by reformulating $\gamma_2$ as a Rayleigh quotient and then we can find an achievable lower bound on the received SNR. Specifically, we can view the HAP as a virtual relay in the second hop and define $y_0 = \sqrt{p_t}{\hat {\bf f}}_0^H {\bf w}_2$ to account for the direct link. Then, we can rewrite $\gamma_2$ in~\eqref{equ:gamma} as $\gamma_{2} = \frac{({\bf x \circ g})^H ({\bf y} {\bf y}^H) ({\bf x \circ g})}{ ({\bf x \circ g})^H ({\bf x \circ g}) }$, where ${\bf x}$ and ${\bf g}$ are $(N+1)\times 1$ vectors and we require $x_0 g_0 = 1$. The symbol $\circ$ denotes the Hadamard product. Let ${\bf z} = {\bf x \circ g}$ and then we have $\gamma_2 = {\bf z}^H ({\bf y} {\bf y}^H){\bf z}/||{\bf z}||^2$, which implies the following inequality
\begin{equation}\label{equ-ub}
\gamma_1 +\gamma_2 \leq \bar\gamma({\bf w}_1) \triangleq p_t||\hat {\bf f}_0||^2 + p_t\sum_{n\in\mathcal{N}_a\cup\{0\}} |\hat {\bf f}_n^H{\bf w}_1|^2.
\end{equation}
The first inequality holds due to the property of the Rayleigh quotient and it holds with equality when ${\bf z} = c {\bf y}$ for some scalar $c$, which implies the following equality constraints:
\begin{subequations}\label{equ_cond_x0g0}
\begin{align}
{p_t c^2| \hat{\bf f}_0^H {\bf w}_2|^2} &= 1, \text{ and }\label{equ_cond1 x0g0}\\
{p_t c^2 }|\hat{\bf f}_n^H {\bf w}_1|^2 &= {\frac{p_n|\hat g_n|^2 }{1 +  p_t|\hat {\bf f}_n^H {\bf w}_1|^2}} , \quad \forall\, n \in \mathcal{N}.\label{equ_cond2 x0g0}
\end{align}
\end{subequations}
Hence, we can find the lower bound on~\eqref{prob_ts} by maximizing~$t\log_2(1+\bar\gamma({\bf w}_1))$ subject to the above two constraints as well as the power budget constraint in~\eqref{con_ts_power}. Note that ${\bf w}_2$ can be aligned with the direct channel $\hat{\bf f}_0$. Thus, we can rewrite~\eqref{equ_cond1 x0g0} as $p_t c^2 = ||\hat{\bf f}_0||^{-1} $ and then reformulate~\eqref{equ_cond2 x0g0} as:
\begin{equation}\label{equ_cond}
|\hat{\bf f}_n^H {\bf w}_1|^2 =\frac{p_n g_n^2||\hat{\bf f}_0||^2}{1 + p_t |\hat{\bf f}_n^H {\bf w}_1|^2} ,  \quad \forall\, n\in\mathcal{N},
\end{equation}
which can be further substituted into~\eqref{con_ts_power}, resulting in the lower bound in~\eqref{prob_trans_ts}.
\end{proof}
Problem~\eqref{prob_trans_ts} is still non-convex due to the coupling between $t$ and ${\bf w}_1$. Though convexity is not assured, monotonicity is also an appealing structural property that can be exploited for efficient algorithm design~\cite{monot13}. Specifically, by a change of variable, the SNR $\bar \gamma$ can be viewed as the decision variable. Then we can rewrite~\eqref{obj_trans_ts} in a simpler form as $r(t,\bar\gamma) = t \log_2(1 + \bar\gamma)$, subject to the feasible set as follows:
\begin{equation}\label{equ-normal}
\Omega \triangleq \left\{
(t,\bar\gamma) \left| \begin{array}{l}\bar\gamma \leq  p_t||\hat{\bf f}_0||^2 +  p_t \sum_{n=0}^N s_{n,1}, \eqref{cons_trans_ts}-\eqref{con2_trans_ts},\\ s_{n,i}\leq |\hat{\bf f}_n^H{\bf w}_i|^2, \text{for }i\in\{0,1\}, n\in{0}\cup \mathcal{N}_a.\end{array}\right.\right\}
\end{equation}
It is clear that the new objective $r(t,\bar\gamma)$ is monotonically increasing in both $t$ and $\bar\gamma$. This implies that its optimum will be achieved on the boundary of the feasible set $\Omega$.

\subsection{Structural Property of the Feasible Region}
The monotonic optimization algorithm successively approximates the feasible set $\Omega$ by regularly-shaped polyblocks~\cite{monot13}, which is a union of finite box sets. A box set is in the form of $[{\bf 0}, {\bf v}]$, where ${\bf v}$ is the end point of the box, and also called a vertex of the polyblock. Let polyblock $P$ be an approximation of the feasible set $\Omega$ and we have $P\supset\Omega$. As $r(t,\bar\gamma)$ is monotonically increasing, its upper bound can be obtained on one of the vertex points of the polyblock $P$, i.e.,
\begin{equation}\label{equ-approx}
\max_{(t,\bar\gamma) \in {  \Omega} } r(t,\bar\gamma) \leq \max_{(t,\bar\gamma) \in P} r(t,\bar\gamma) = \max_{(t,\bar\gamma) \in V} r(t,\bar\gamma),
\end{equation}
where $V$ denotes the set of vertices of the polyblock $P$. The first inequality holds due to the fact that ${\Omega} \subset P $ while the equality holds due to the monotonicity of the objective function $r(t,\bar\gamma)$. Though the maximization of $r(t,\bar\gamma)$ over a non-convex set $\Omega$ is difficult, its upper bound can be easily obtained by evaluating the objective function on each of the vertex points in set $V$.

To approximate the optimum $\max_{(t,\bar\gamma) \in { \Omega} } r(t,\bar\gamma)$, the above inequality in~\eqref{equ-approx} implies that we need to generate smaller polyblocks successively to approximate the feasible region ${ \Omega}$, which depends on the structural property of $\Omega$. To proceed, we can verify that the feasible set $\Omega$ defined in~\eqref{equ-normal} represents a normal set, which bears the following structural property:
\begin{proposition}\label{pro_normal_ts}
Given $(t, \bar\gamma) \in \Omega$, we always have $(t', \bar\gamma')\in\Omega$ for any $(t', \bar\gamma') \preceq (t, \bar\gamma)$.
\end{proposition}
\begin{proof}
Suppose that $(t, \bar\gamma) \in \Omega$ and the feasible solution is given by $(t,{\bf w}_0,{\bf w}_1)$. To show $(t', \bar\gamma')\in\Omega$, we need to construct a new solution $(t',{\bf w}_0',{\bf w}_1')$ such that~\eqref{cons_trans_ts}-\eqref{con2_trans_ts} hold with $(t', \bar\gamma')$. For this purpose, we can set $\alpha =  \bar\gamma'/\bar\gamma$ and then construct the new solution as ${\bf w}_1' = \alpha {\bf w}_1$ and ${\bf w}_0' = \alpha {\bf w}_0$. This implies that
\begin{align*}
s_{n,1}'( 1 + p_t s_{n,1}') ~& \leq~ \alpha s_{n,1}(1 + p_t s_{n,1} ) \leq \alpha \psi_n  ( 1/t -2  ) s_{n,0}\leq \psi_n ( 1/t' -2  )s_{n,0}'.
\end{align*}
The first inequality is due to $\alpha \leq 1$ and the second inequality holds as $(t, \bar\gamma) \in \Omega$ with the solution $(t,{\bf w}_0,{\bf w}_1)$. The last inequality easily follows as we have $t' \leq t$. Hence, the new solution $(t',{\bf w}_0',{\bf w}_1')$ is feasible to~\eqref{cons_trans_ts}-\eqref{con2_trans_ts}. Besides, we have $\bar \gamma' = \alpha \bar \gamma \leq  p_t||\hat{\bf f}_0||^2 +    p_t \sum_{n=0}^N s'_{n,1}$. This implies that $(t', \bar\gamma') \in \Omega$.
\end{proof}


The structural property in Proposition~\ref{pro_normal_ts} will help find a smaller polyblock $P'$ that gives a closer approximation to $\Omega$. The generation of new polyblocks follows a similar procedure as that in~\cite{yztcom}. Specifically, if we have some ${\bf z}\in P$ and ${\bf z}\notin \Omega$, it follows that ${\bf z}'\notin\Omega$ for any ${\bf z}'\succeq {\bf z}$ by Proposition~\ref{pro_normal_ts}. This implies that we can construct a better approximation $P'$ by cutting off a subset from $P$, i.e.,
\begin{equation}\label{equ_newblock}
P' = P \setminus \{{\bf z}'\in P: {\bf z}'\succeq {\bf z}\}.
\end{equation}
This ensures that $P\supset P' \supset \Omega$. The initial polyblock $P_0$ can be simply set as the box set that covers the whole feasible set $\Omega$. Hence, the vertex of the polyblock $P_0$ is given by $(1/2, \bar\gamma_{\max})$, where $\bar\gamma_{\max}$ denotes a upper bound on $\bar\gamma$. From~\eqref{equ-normal}, we can simply set $\bar\gamma_{\max}$ as follows:
\[
\bar\gamma_{\max} =  2 p_t||\hat {\bf f}_0||^2 + p_t\sum_{n\in\mathcal{N}_a} ||\hat {\bf f}_n||^2.
\]

In the $k$-th iteration of the polyblock approximation, we first determine an upper bound $r_k^U$ on one vertex of the polyblock $P_k$, i.e., $r_{k}^{U}=\max _{\mathbf{z} \in V_{k}} r(\mathbf{z})$. Then we have the optimal vertex point $\mathbf{z}_{k}=\arg \max _{\mathbf{v} \in V_{k}} r(\mathbf{v})$ and the corresponding upper bound $r_{k}^{U}=r\left(\mathbf{z}_{k}\right)$. A lower bound $r_k^L$ can be determined by projection, detailed in the next part. During successive polyblock approximation, the update to $r_k^L$ and $r_k^U$ will become closer to each other in each iteration. The iteration terminates when the gap between $r_k^L$ and $r_k^U$ is below a pre-defined error bound~\cite{yztcom}.

\subsection{Projection via Solving SDPs}

The main challenge of the monotonic optimization algorithm lies in the search for the lower bound $r_k^L$ in each iteration. Fortunately, we show that the lower bound $r_k^L$ can be achieved optimally by solving a set of semidefinite programs (SDPs). Specifically, if the optimal vertex point $\mathbf{z}_{k}$ is feasible to $\Omega$, we can conclude that $r_k^L = r_k^U$ and then terminate the algorithm. However, if $\mathbf{z}_{k}\notin\Omega$, we can multiply $\mathbf{z}_{k}$ by a scaling factor $\lambda\in(0,1)$ to project the infeasible $\mathbf{z}_{k}$ onto the boundary of $\Omega$. The projection aims at finding the maximum factor $\lambda_{k}$ such that $\lambda_k \mathbf{z}_k \in \Omega$, i.e., $\lambda_{k} = \arg\max\{\lambda: (\lambda t_{k}, \lambda \bar{\gamma}_{k}) \in\Omega \}$. It is easy to see that $\lambda \mathbf{z}_k \in \Omega$ for $\lambda \in(0, \lambda_{k} ]$ and $\lambda \mathbf{z}_k \notin\Omega$ for $\lambda \in(\lambda_{k} , 1]$, by the structural property in Proposition~\ref{pro_normal_ts}. This implies a bisection method to find the maximum $\lambda_{k}$ and the projection point $\mathbf{o}_k \triangleq \lambda_k \mathbf{z}_k$. Given a fixed $\lambda\in(0,1)$, we can increase it in next iteration if $\lambda \mathbf{z}_k \in \Omega$, otherwise decrease it by the bisection method~\cite{yztcom}.

By the definition of $\Omega$ in~\eqref{equ-normal}, the feasibility check $(\lambda t_{k}, \lambda \bar{\gamma}_{k}) \in \Omega$ with a fixed $\lambda$ is equivalent to solve the following problem:
\begin{subequations}\label{prob_feas}
\begin{align}
\max_{t,{\bf w}_0,{\bf w}_1}~&~ p_t||\hat {\bf f}_0||^2 + p_t\sum_{n\in\mathcal{N}_a\cup\{0\}} s_{n,1} \label{obj_feas}\\
s.t.~&~ \left[\begin{matrix} \psi_n \Big(\frac{1}{\lambda t_k} - 2 \Big) s_{n,0} - s_{n, 1}, & \sqrt{p_t} s_{n,1} \\
\sqrt{p_t} s_{n,1}, & 1\end{matrix}\right] \succeq 0 , \label{con_feas_matrix}\\
~&~ s_{n,i} \leq |\hat{\bf f}_n^H{\bf w}_i|^2 \text{ and } ||{\bf w}_i|| \leq 1, \label{con_feas_w} \\
~&~ i\in\{0,1\} \text{ and } n\in{0}\cup \mathcal{N}_a . \label{con_feas_n}
\end{align}
\end{subequations}
The constraint in~\eqref{con_feas_matrix} is an equivalence of~\eqref{cons_trans_ts} and it becomes a linear matrix inequality. Considering semidefinite relaxation (SDR), we can transform~\eqref{prob_feas} into a convex form and solve it efficiently by semidefinite programming. Let $m_k$ denote the optimum to~\eqref{prob_feas} and then we can conclude that $(\lambda t_{k}, \lambda \bar{\gamma}_{k}) \in \Omega$ if $m_k\geq \lambda \bar \gamma_k$. Given the projection point $\mathbf{o}_k=\lambda_k \mathbf{z}_k$, the construction of new polyblock $P_{k+1}$ follows a similar procedure as in~\eqref{equ_newblock}. In particular, we can construct a separating cone $P_{k}^{c} \triangleq\left\{\mathbf{z} | \mathbf{z} \succeq\mathbf{o}_k\right\}$ such that $P_{k}^{c} \cap \Omega=\emptyset$. Then a new polyblock $P_{k+1}$ can be generated by cutting off $\Delta_{k} \triangleq P_{k}^{c} \cap P_{k}=\left\{\mathbf{z} \in P_{k} | \mathbf{z} \succeq\mathbf{o}_k\right\}$ from the polyblock $P_{k}$. The detailed procedures can be referred to~\cite{yztcom}.

\section{Optimization-driven Hierarchical DRL Framework}\label{sec_H-DDPG}


The optimization methods are typically based on a simplified system model, e.g., with perfect channel information or ideal system implementation. Though lots of efforts are devoted to designing exact solutions, the optimization based on a simplified system model only provides lower bounds or approximations to the original problems. In this following, we propose a novel optimization-driven DRL framework for throughput maximization in a hybrid relay network, by exploiting the efficiency of optimization methods and the robustness of DRL approaches. In particular, we use DRL to build the learning framework that is robust to complex problem structure and inexact modeling, while integrate optimization methods in the inner loop to reduce the search space and improve the learning efficiency.

\subsection{MDP Reformulation of Problem~\eqref{prob_bin}}

DRL is a combination of DNNs and RL, aiming at solving MDP problems with large action and state spaces. The most straightforward DRL solution to problem~\eqref{prob_bin} is to reformulate it into an MDP and design a single agent at the HAP, which jointly decides the HAP's beamforming and the relaying strategies simultaneously based on the observed state ${\bf s}_t\in\mathcal{S}$ and the knowledge learnt from past {experiences}. The system state ${\bf s}_t = ({\bf c}_t, {\bf e}_t)$ at the $t$-th decision epoch is a combination of the channel conditions ${\bf c}_t$ and the relays' energy status ${\bf e}_t = [e_{1,t},e_{2,t},\ldots, e_{N,t}]^T$. The channel conditions include the direct channel $\hat{\bf f}_0$ from the HAP to the receiver and each relay's channel $(\hat{\bf f}_n, \hat{g}_n)$, for $n\in\mathcal{N}$. Hence, we denote ${\bf c}_t = \{\hat{\bf f}_0, (\hat{\bf f}_n, \hat{g}_n)_{n\in\mathcal{N}}\}$. As observed from~\eqref{equ_channel_direct}-\eqref{equ_channel_relay}, the channel conditions depend on the relays' mode selection $b_n$ and operating parameters $\Gamma_k$. The available energy $e_{n,t}$ at relay-$n$ includes the initial residual energy $e_{n,t-1}$ and the harvested energy $h_n\triangleq \eta(1-2t)p_t|\hat{\bf f}_n^H{\bf w}_0|^2$ from RF signals, as shown in~\eqref{con_powerbd}. The dynamics of each relay's energy status also depend on the power consumption $p_n$ in relay communications. The power consumption $p_n$ of active relays depends on data rate in signal forwarding, while it is a small constant for passive relays. We assume  that both channel conditions ${\bf c}_t$ and the energy status ${\bf e}_t$ can be measured at the beginning of each decision epoch. Given the current state ${\bf s}_t$, the action ${\bf a}_t = (t, {\bf w}_0, {\bf w}_1, {\bf o}_t)$ of the DRL agent includes the HAP's the time allocation $t$ and beamforming strategies $({\bf w}_0, {\bf w}_1)$ for energy and information transfer. ${\bf o}_t = \{ b_n, \theta_n\}_{n\in\mathcal{N}}$ denotes each relay's binary mode selection $b_n\in\{0, 1\}$ and the operating parameter $\theta_n$. Except the binary mode selection $b_n$, all other decision variables are continuous variables. We aim to maximize the throughput in problem~\eqref{prob_bin}, subject to the relays' power budget constraints. Hence, the immediate reward can be simply set as the objective in problem~\eqref{prob_bin}, i.e.,~$v_t({\bf s}_t, {\bf a}_t) = t\log_2{(1+\gamma_1+\gamma_2)}$.

\subsection{Model-free DQN and DDPG Algorithms}


By Bellman equation, we can simplify the policy optimization as the optimization of action ${\bf a}_t$ in an iterative equation~\cite{sutton}: $V^{\pi^*}({\bf s})
=\max_{{\bf a}\in\mathcal{A}} v_t({\bf s}_t, {\bf a}) + \gamma \mathbb{E}\left[ V({\bf s}_{t+1})\right]$, where $\gamma\in(0,1)$ denotes a discount factor. The expectation is taken over all possible state transitions from the current state ${\bf s}_t$ to the next state ${\bf s}_{t+1}$ when taking action ${\bf a}$ in state ${\bf s}_t$. With small and finite state and action spaces, the optimal policy $\pi^*$ can be obtained by the $Q$-learning algorithm. Specifically, the optimal action at each state is to maximize the $Q$-value, i.e.,~${\bf a}_t^* = \arg\max_{{\bf a}\in\mathcal{A}} Q({\bf s}_t, {\bf a})$, where the $Q$-value is defined as $Q({\bf s}_t, {\bf a}_t)= v_t({\bf s}_t, {\bf a}) + \gamma \mathbb{E}[ V({\bf s}_{t+1})]$. For algorithmic implementation, we can randomly initialize the $Q$-value and then update it by the difference between the current $Q$-value and its target $y_t$, i.e.,
\[
Q_{t+1}({\bf s}_t,{\bf a}_t) = Q_t({\bf s}_t,{\bf a}_t) + \tau_t \Big[ y_t - Q_t({\bf s}_t, {\bf a}_t)\Big],
\]
where $\tau_t$ is a step-size and $y_t$ is evaluated as follows:
\begin{equation}\label{equ_target}
y_t =  r_t({\bf s}_t, {\bf a}_t) + \gamma \max_{{\bf a}_{t+1}} Q_t({\bf s}_{t+1}, {\bf a}_{t+1}).
\end{equation}
The difference $ y_t - Q_t({\bf s}_t, {\bf a}_t)$ is called the temporal-difference (TD) error. For small-size state and action spaces, the $Q$-value can be stored in a table and updated in each decision epoch.

For very large state and action spaces, the DQN algorithm uses DNNs with the weight parameters $\boldsymbol{\omega}_t$ to approximate the $Q$-value function. The input to the DNNs is the current state ${\bf s}_t$ and output is the expected action ${\bf a}_t$. The training of DNNs aims at minimizing a loss function:
\begin{equation}\label{equ_loss}
\ell(\boldsymbol{\omega}_t)=\mathbb{E}\left[(y_{i} -Q_{t}({\bf s}_i, {\bf a}_i| \boldsymbol{\omega}_t))^{2}\right].
\end{equation}
Comparing to the conventional $Q$-learning algorithm, DQN improves the learning efficiency by using a set of historical transition samples $({\bf s}_i, {\bf a}_i, r_i, {\bf s}_{i+1})\in\mathcal{M}_t$, namely, a mini-batch, to train the DNN parameters ${\bm \omega}_t$ at each decision epoch $t$. The target value $y_i$ in the loss function~\eqref{equ_loss} is evaluated by~\eqref{equ_target} for each sample and the expectation is taken over all samples in the mini-batch $\mathcal{M}_t$. Besides, DQN algorithm ensures more stable learning by using a double $Q$-network structure. The online $Q$-network with DNN parameters ${\bm \omega}_t$ generates the value estimation $Q_{t}({\bf s}_i, {\bf a}_i| \boldsymbol{\omega}_t)$ given the state-action pair $({\bf s}_i, {\bf a}_i)$, while the target value $y_i$ is generated by the target $Q$-network with a different set of DNN parameters ${\bm \omega}_t'$, which are delayed copies from the online $Q$-network, i.e.,~${\bm \omega}_t' = {\bm \omega}_{t-t_d}$, where $t_d$ denotes the time delay to the current decision epoch.

A similar double $Q$-network structure also appears in the DDPG algorithm~\cite{ddpg}, which extends the DQN algorithm to solve optimization problems with continuous action space. Besides the DNN approximation for the $Q$-value, DDPG also approximates the policy $\pi_{\bm \mu}$ by using DNNs. The DNN training aims at updating the DNN parameters ${\bm \mu}$ in a gradient direction to improve the value function, which can be rewritten as follows
\[
J({\bm \mu}) = \sum_{{\bf s} \in \mathcal{S}} p({\bf s}) \sum_{{\bf a} \in \mathcal{A}} \pi_{\bm \mu}({\bf a} \vert {\bf s}) Q({\bf s},{\bf a}|\boldsymbol{\omega}),
\]
where $p({\bf s})$ denotes the stationary state distribution corresponding to the policy $\pi_{\bm \mu}$ and $Q({\bf s},{\bf a}|\boldsymbol{\omega})$ is the parameterized $Q$-value. Deterministic policy gradient theorem in~\cite{ddpg} simplifies the gradient evaluation as
\[
\nabla_{\bm \mu} J({\bm \mu}) = \mathbb{E}_{{\bf s} \sim p({\bf s})} [ \nabla_{\bf a} Q({\bf s},{\bf a}|\boldsymbol{\omega}) \nabla_{\bm \mu} \pi_{\bm \mu}({\bf s}) \rvert_{{\bf a}=\pi_{\bm \mu}({\bf s})}],
\]
which can be performed efficiently by sampling the historical trajectories. The policy gradient $\nabla_{\bm \mu} J({\bm \mu})$ motivates the actor-critic framework in~\cite{ddpg}, which updates the DNN parameters $({\bm \mu},{\bm \omega})$ separately.
The actor network updates the policy parameters ${\bm \mu}$ in gradient direction as follows:
\[
{\bm \mu}_{t+1} = {\bm \mu}_t + \alpha_{\mu} {\nabla_{\bf a} Q({\bf s}_t, {\bf a}_t|{\bm \omega}_t) \nabla_{\bm \mu} \pi_{\bm \mu}({\bf s}) \rvert_{{\bf a}_t=\pi_{\bm \mu}({\bf s})}}.
\]
Similar, the critic network updates the $Q$-network as follows:
\[
{\bm \omega}_{t+1} = {\bm \omega}_t + \alpha_{\omega} \delta_t \nabla_{{\bm \omega}} Q({\bf s}_t, {\bf a}_t|{\bm \omega}_t),
\]
where $\delta_t  = y_t - Q({\bf s}_t, {\bf a}_t|{\bm\omega}_t)$ denotes the TD error between $Q({\bf s}_t, {\bf a}_t|{\bm\omega}_t)$ and its target value $y_t$. Two constants $\alpha_{\mu}$ and $\alpha_{\omega}$ are viewed as step-sizes. Similar to DQN algorithm, the training of critic network is performed by sampling a mini-batch of samples from the experience replay memory, aiming to minimize the loss function in~\eqref{equ_loss}, where the target value $y_i$ is given by
\[
y_i = v_i({\bf s}_i, {\bf a}_i) + \gamma Q({\bf s}_{i+1},\pi({\bf s}_{i+1}|{{\bm \mu}'_t})|{\bm \omega}'_t).
\]
The DNN parameters $({{\bm \mu}'_t}, {\bm \omega}'_t)$ of the target networks are regular copies from their online networks $({{\bm \mu}_t}, {\bm \omega}_t)$, respectively.

\subsection{Optimization-driven Learning Strategy}

From the above inspection, we observe that either the DQN or the DDPG algorithm relies on periodical parameters copying from the online $Q$-network to the target $Q$-network. This implies strong coupling between the online and target $Q$-networks and may lead to slow learning rate and unstable issues. As both the online and target $Q$-networks are randomly initialized in the DQN or DDPG algorithm, the evaluation of the immediate reward $v_t({\bf s}_t, {\bf a}_t)$ can be far from its real value in the early stage of learning, which probably misleads the learning process. Therefore, we require a long warm-up period to train both $Q$-networks of the DQN or DDPG algorithm. Besides, the parameters copying from the online $Q$-network to the target $Q$-network is critical to the learning performance. Frequent parameters copying implies unstable learning and even divergence, while less frequent copying slows down the convergence rate. Thus, the optimal setting for parameters copying becomes problematic for practical implementation.

Considering the above difficulties, in the sequel, we aim to stabilize and improve the learning efficiency by proposing two special designs detailed as follows.

\subsubsection{Hierarchical Integration of DQN and DDPG}

Note that the original design problem in~\eqref{prob_bin} is a mixed optimization problem with both discrete and continuous decision variables. The binary variable $b_n$ determines the relay's operating mode, while the continuous variables include the the HAP's beamforming $({\bf w}_0, {\bf w}_1)$ and time allocation strategy, as well as the passive relays' reflecting phases ${\bm \theta}= [\theta_1, \theta_2, \ldots, \theta_N]^T$. The basic idea of the hierarchical learning framework is to split the action vector into two parts $({\bf b}_t, {\bf a}_t^c)$. The discrete action ${\bf b}_t = [b_1, b_2, \ldots, b_N]^T $, indicating the relays' operating modes, can be learnt following the conventional DQN algorithm. Given the discrete action ${\bf b}_t$ in the outer-loop DQN algorithm, we can focus on the solution ${\bf a}_t^c = (t, {\bf w}_0, {\bf w}_1, {\bm \theta})$ to a continuous optimization problem in~\eqref{prob_ts}, which can be considered in the DDPG framework. Such a hierarchical DDPG (denoted as H-DDPG) framework allows us to decompose the combinatorial and discrete optimization of the relays' radio modes from the optimization of other continuous variables. The convergent value function of the inner-loop DDPG algorithm can be viewed as the $Q$-value of the outer-loop DQN. The benefit of the proposed H-DDPG framework mainly lies in that it reduces the action space, and thus it is expected to improve the learning efficiency.

\subsubsection{Optimization-driven DDPG Algorithm}
Besides the hierarchical design, we also upgrade the inner-loop DDPG algorithm for solving the continuous control problem in~\eqref{prob_ts}. Note that the conventional DDPG algorithm is subject to slow convergence speed due to random initialization of the double $Q$-networks. To improve the learning efficiency, we propose a model-based optimization method to provide a better-informed target value for DNN training. Specifically, in the $t$-th decision epoch, the actor network outputs an action ${\bf a}_t^c = (t,{\bf w}_0, {\bf w}_1 ,{\bm \theta})$ and the target $Q$-network produces the target value $y_t$. Note that $y_t$ can be very different from its optimal value. We employ the optimization framework developed in Section~\ref{sec_H-DDPG} to approximate problem in~\eqref{prob_ts}, which can be used to estimate a lower bound on the target value $y_t$. Let $y_t'$ denote the optimization-driven target value and ${\bf a}_t^o = (t',{\bf w}_0', {\bf w}_1' ,{\bm \theta}')$ denote the approximate solution in the optimization framework. We can expect that the model-based optimization can provide a more accurate estimation of the target $Q$-value compared to the DNN generated target value $y_t$, especially in the early stage of learning. A better-informed target value $y_t'$ can help the inner-loop DDPG algorithm adapt faster and achieve a better reward performance. Moreover, the derivation of the optimization-driven target value $y_t'$ is irrelevant with the online $Q$-network. This implies that $y_t'$ is more stable than the output $y_t$ from the target $Q$-network. Such a decoupling between the online $Q$-network and its target is expected to reduce the performance fluctuations and stabilize the learning performance in a shorter training time.

\begin{figure}[t]
\centering
\includegraphics[width=0.95\textwidth]{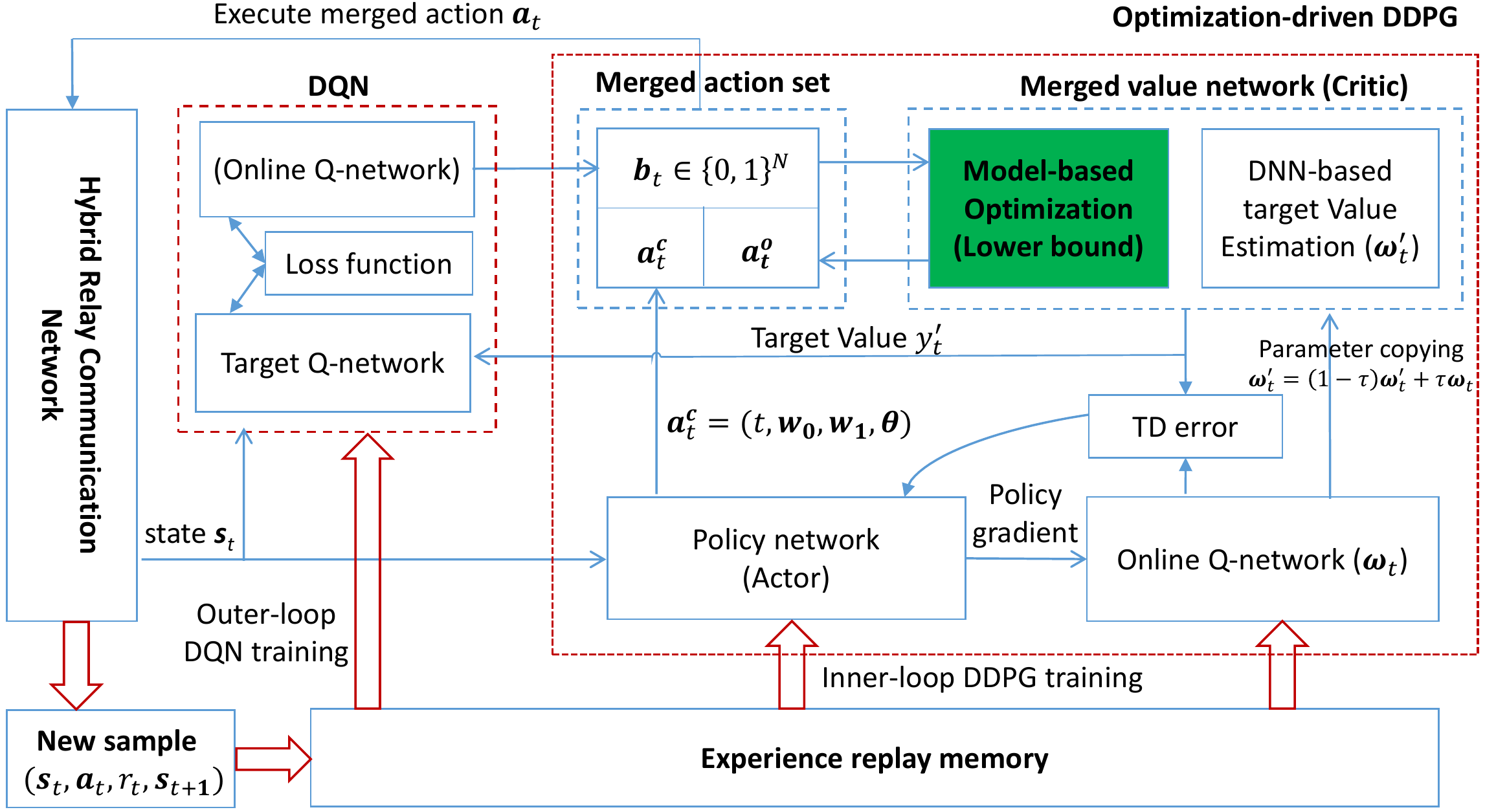}
\caption{The optimization-driven hierarchical DDPG (H-DDPG) framework for hybrid relaying communications.}\label{fig:H-DDPG}
\end{figure}

Figure~\ref{fig:H-DDPG} shows the main building blocks of the proposed optimization-driven H-DDPG algorithm. The complete framework contains two parts, the outer-loop DQN and the inner-loop model-based DDPG. Given the system state, the outer-loop DQN firstly decides the radio mode ${\bf b}_t$ for each relay, and then the channel conditions can be updated according to the relays' radio modes. Based on the channel conditions and the transmission performance, the actor and critic networks of the inner-loop DDPG algorithm generate the continuous action ${\bf a}_t^c$ and the value estimate, respectively. The action ${\bf a}_t^c$ can be merged with the optimized solution ${\bf a}_t^o$ based on the values of $y_t$ and $y_t'$. In the simplest case, we use $y_t'$ as the target value for DNN training if $y_t' > y_t$ and meanwhile accept the new action~${\bf a}_t = ({\bf b}_t,{\bf a}^o_t)$ instead of ${\bf a}_t = ({\bf b}_t,{\bf a}^c_t)$. Note that we may also have $y_t' < y_t$ when the learning becomes more stable. As such, we follow exactly the output of the actor network. When the inner-loop DDPG becomes stable, the target value $y'_t$ of the DDPG algorithm is also used as a reference or benchmark for the DQN algorithm to improve its learning speed. The detailed algorithm is illustrated in Algorithm~\ref{alg:hddpg}.

\begin{algorithm}[t]
    \caption{Model-based H-DDPG Algorithm for Hybrid Relaying Communications}\label{alg:hddpg}
    \begin{algorithmic}[1]
        \State Randomly initialize $Q$-networks of DQN and DDPG
        \State Initialize the replay buffer for DQN and DDPG
        \State \textbf{DQN loop}
        \State \hspace{5mm} Select the relays' radio modes ${\bf b}$ by DQN algorithm
        \State \hspace{5mm} \textbf{DDPG loop}
        \State \hspace{10mm} Generate continuous action ${\bf a}_t^c$ by DDPG algorithm
        \State \hspace{10mm} Randomly exploit action ${\bf a}_t^c \leftarrow {\bf a}_t^c + \mathcal{N}_0$
        \State \hspace{10mm} Evaluate the reward $v_t$ of the action ${\bf a}_t^c$
        \State \hspace{10mm} Update $y_t'$ and ${\bf a}_t^o$ by the optimization method
        \State \hspace{10mm} \textbf{If} $y_t' > y_t$
        \State \hspace{15mm} Accept optimization-driven target $y_t\leftarrow y_t'$
        \State \hspace{15mm} Adopt the merged action ${\bf a}_t \leftarrow ({\bf b}_t , {\bf a}_t^o )$
        \State \hspace{10mm} \textbf{END If}
        \State \hspace{10mm} Buffer transition sample in DDPG memory
        \State \hspace{10mm} Sample a minibatch from DDPG memory
        \State \hspace{10mm} Update critic and actor networks
        \State \hspace{5mm} \textbf{END DDPG loop}
        \State \hspace{5mm} Sample a minibatch from DQN memory
        \State \hspace{5mm} Update DQN network and its target network
        \State \textbf{END DQN loop}
    \end{algorithmic}
\end{algorithm}

\subsection{Further Discussions}\label{subsec_dis}

The optimization module in Fig.~\ref{fig:H-DDPG} aims to provide a lower bound estimation of the target $Q$-value, based on incomplete or inaccurate system information. For example, we can explore the physical connections between different control variables. The control variables of a complex problem are usually high dimensional and untractable jointly by optimization methods. However, given one part of the control variable, we can estimate the other part efficiently by solving an approximate and usually convex optimization problem. This approximate optimization problem can be built based on the physical connections between different control variables.

We require that the optimization solution in the learning framework does not need to be accurate, but to be very efficient. The accuracy of the solution usually means more complicated computations and a lot of effort in algorithm design, which may delay the learning in each episode. Therefore, besides the monotonic optimization for problem~\eqref{prob_trans_ts}, we also try to approximate problem~\eqref{prob_trans_ts} in a much simplified case with a fixed time variable $t$. In particular, when the inner-loop DDPG algorithm generates the action ${\bf a}_t^c = (t,{\bf w}_0, {\bf w}_1 ,{\bm \theta})$, the optimization module in Fig.~\ref{fig:H-DDPG} takes the variable $t$ as input and optimizes the other part of the control variables $({\bf w}_0, {\bf w}_1 ,{\bm \theta})$. With fixed $t$, the optimization problem in~\eqref{prob_trans_ts} becomes an SDP in the form of problem~\eqref{prob_feas}, which can be solved very efficient by the interior-point algorithms~\cite{rankluo}. Note that the monotonic optimization algorithm requires an iterative procedure and the SDP in~\eqref{prob_feas} will be solved multiple times in each iteration. Therefore, the simplified case with fixed $t$ can significantly reduce the time spent on the inner-loop optimization, therefore improve the overall learning efficiency. Our numerical results in Section~\ref{sec_sim} reveal that such a model-simplified H-DDPG algorithm still outperforms the conventional model-free H-DDPG algorithm significantly. Most importantly, it has very close performance as that of the H-DDPG algorithm driven by monotonic optimization. This verifies that the our learning framework is robust to the optimization methods (i.e.,~insensitive to different accuracies), and can achieve significant performance gain when partial system information is considered in the inner-loop optimization.

The practical implementation of the proposed optimization-driven H-DDPG algorithm follows the system structure as shown in Fig.~\ref{fig:topo}. The hierarchial learning agent is deployed at the HAP, which makes decisions based on the perceived transmission performance, the relays' channel and energy conditions. The channel conditions can be viewed as block fading and kept constant during each transmission period~\cite{gsmwcnc19}. In each decision epoch, the HAP and the agent collect all channel information and the relays' energy status to adapt the beamforming and relaying strategies, which will be distributed to the relays via a downlink control channel. The estimation of the throughput performance is based on the feedback information from the receiver. The HAP with more computational resources also implements an optimization module that estimates the optimal strategy and the performance lower bound based on incomplete system information.

\begin{figure*}[t]
    \centering
    \includegraphics[width=0.95\textwidth]{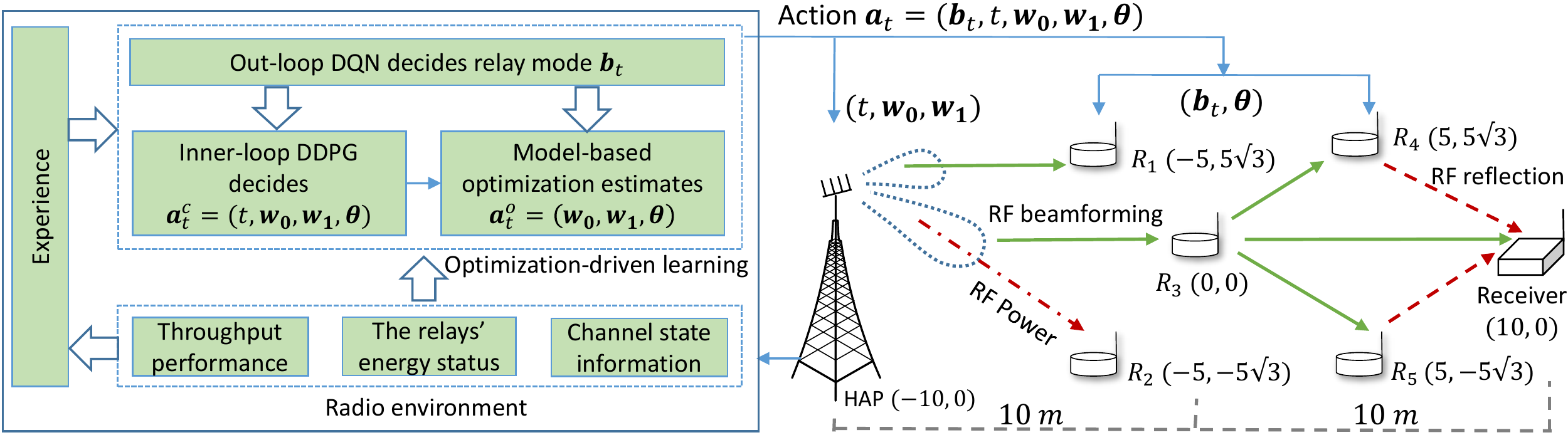}
    \caption{Simulation system.}\label{fig:topo}
\end{figure*}

\section{Numerical Results}\label{sec_sim}


In the simulation, we consider a two-hop relay-assisted transmission system as shown in Fig.~\ref{fig:topo}. We consider 3 antennas at the HAP and 5 relays in the system. The location of each node is depicted in Fig.~\ref{fig:topo}. The hybrid relays can harvest RF energy from the HAP's signal beamforming. The energy harvesting efficiency is set to $0.6$. The HAP's transmit power ranges from -10 dBm to 10 dBm. We consider a log-distance path loss model. The path loss at unit distance is 25 dB and the path loss exponent is set to $2$. The noise power is $-80$ dBm. The reflection coefficients of the passive relays are set to $0.5$. A similar setting has been used in~\cite{hddpg}.

\begin{figure}[t]
    \centering
    \subfloat[Reward performance]{\includegraphics[width=\singlesize\textwidth]{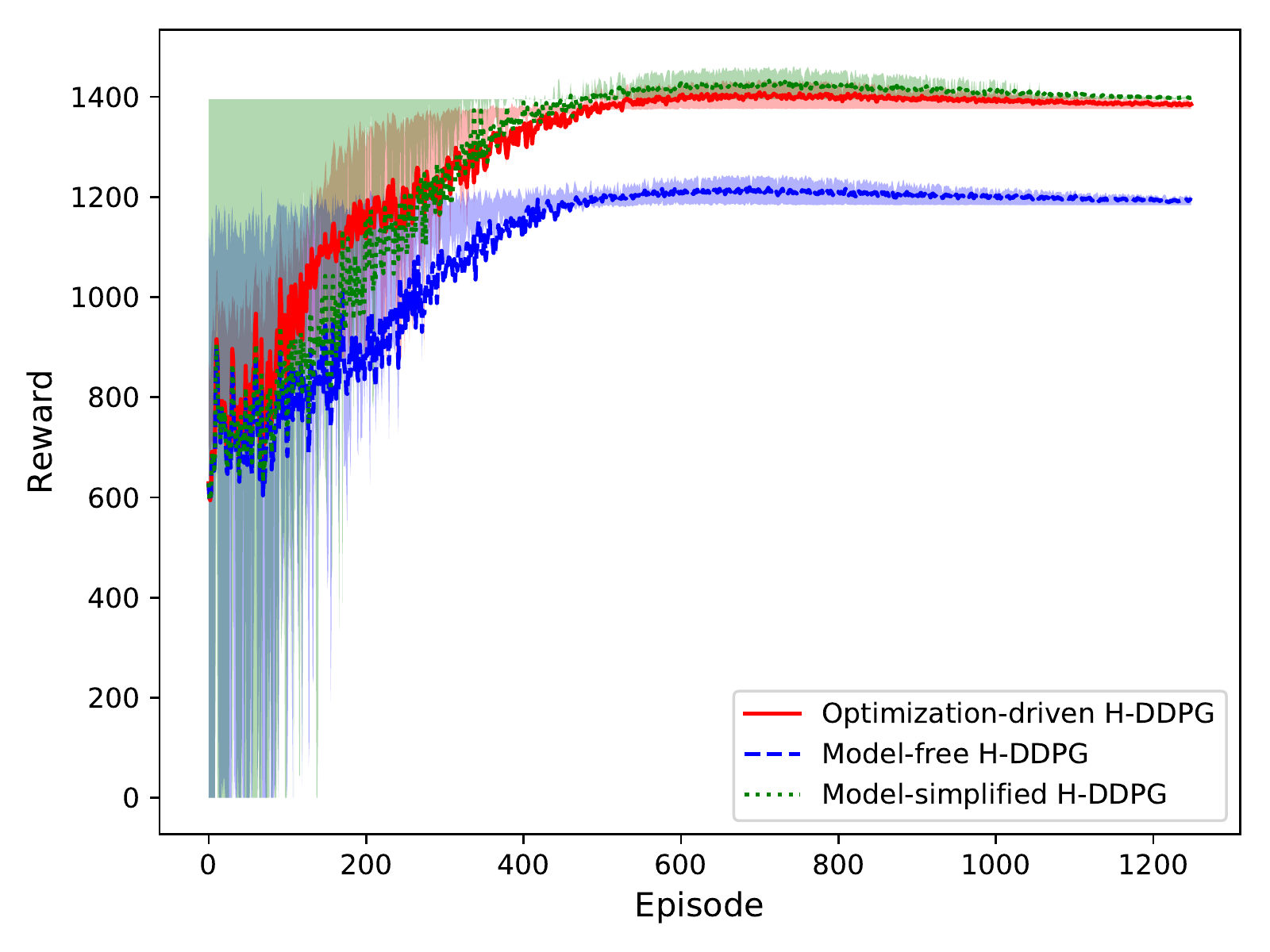}}\\
    \subfloat[Stability in learning]{\includegraphics[width=\singlesize\textwidth]{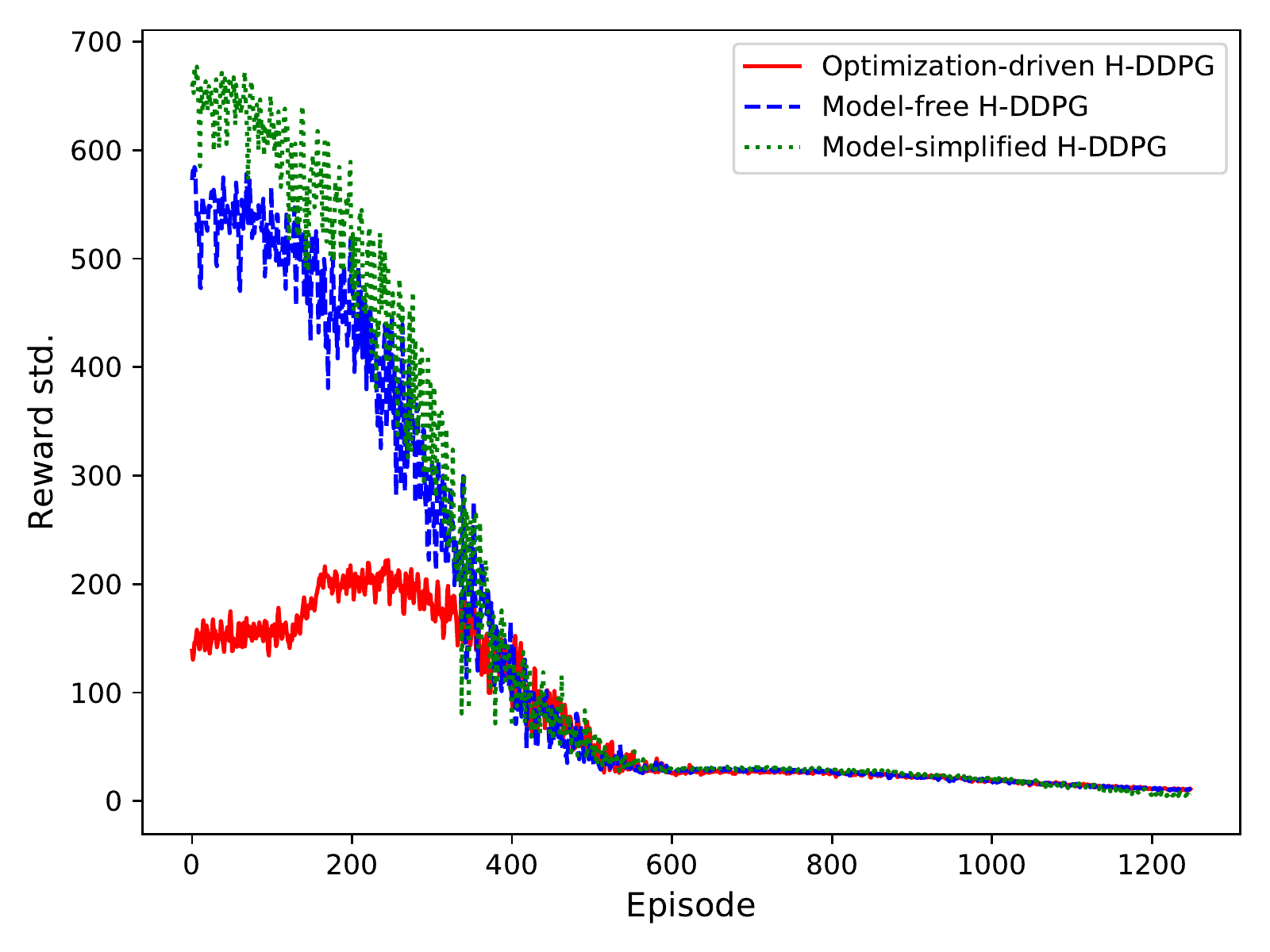}}
    \caption{The comparisons of convergence speed, reward performance, and stability in different H-DDPG algorithms.}\label{fig:drl_compare}
\end{figure}

\subsection{Convergence of the Optimization-driven H-DDPG Algorithm}
We firstly compare the reward performance of the optimization-driven H-DDPG algorithm with that of the conventional model-free H-DDPG algorithm, as shown in Fig.~\ref{fig:drl_compare}(a). The solid lines represent the averaged rewards over $80$ episodes and the shaded areas denote the rewards' fluctuating ranges. The optimization-driven H-DDPG algorithm tries to improve the learning efficiency by using the monotonic optimization method to find a performance lower bound on problem~\eqref{prob_bin} as a better-informed estimation of the target $Q$-value. Based on the discussion in Section~\ref{subsec_dis}, we also implement a model-simplified H-DDPG algorithm that optimizes a part of the control variables $({\bf w}_0, {\bf w}_1 ,{\bm \theta})$ based on the time variable $t$ generated by the inner-loop DDPG algorithm. As such, the optimization module can generate a performance lower bound more efficiently with reduced computational complexity. Table~\ref{tab:exec-time} shows the comparison of the average run time per learning episode in different H-DDPG algorithms. It is clear that the model-simplified H-DDPG algorithm improves the learning efficiency more than 40 times compared to the H-DDPG algorithm driven by the monotonic optimization method. Compared to the model-free H-DDPG algorithm, the average run time is almost tripled in the model-simplified H-DDPG algorithm, which however also achieves a much higher reward performance as shown in Fig.~\ref{fig:drl_compare}(a).

\begin{table}[t]
    \centering
    \caption{Average run time in each episode (in milliseconds)}\label{tab:exec-time}
    \begin{tabular}{|c|r|r|r|r|r|}
    \hline
        Number of hybrid relays & 5 & 4 & 3 & 2 & 1 \\ \hline
        Model-free H-DDPG & 1.25 & 1.24 & 1.22 & 1.20 & 1.19 \\ \hline
        Model-simplified H-DDPG & 5.08 & 4.05 & 3.92 & 3.67 & 3.25 \\ \hline
        Optimization-driven H-DDPG & 227.64 & 214.15 & 199.39 & 185.32 & 168.34 \\ \hline
    \end{tabular}
\end{table}

Though the rewards in three H-DDPG algorithms fluctuate over episodes due to random explorations, finally they converge to stable values after learning from the accumulated past experience. One interesting observation is that the optimization-driven H-DDPG shows a higher convergent reward, which is up to 20\% compared to that of the model-free H-DDPG algorithm. This verifies the advantages of the proposed algorithm. Besides, we observe that the shaded area of the model-based H-DDPG is smaller than that of model-free H-DDPG. This implies that the model-based optimization-driven H-DDPG has a more stable learning performance. This can be further verified in Fig.~\ref{fig:drl_compare}(b) where we show the standard deviation (std.) of the rewards for two H-DDPG algorithms. It is clear that the reward of the optimization-driven H-DDPG has less fluctuations compared to that of the model-free H-DDPG. Another interesting observation is that the variance of the optimization-driven H-DDPG firstly increases and then drops significantly to its minimum. The reason is that the optimization-driven H-DDPG tends to choose the actions given by the optimization module in the early stage of learning. The optimization-driven target provides an independent estimation of the reward and thus has less variance over different episodes. After that, H-DDPG can learn from past experiences and achieve a higher reward, which results in a higher variance due to random exploration.

Comparing two model-based H-DDPG algorithms, we observe that the model-simplified H-DDPG achieves very close performance to that of the optimization-driven H-DDPG algorithm. However, they have very different variances in the learning curves. As shown in Fig.~\ref{fig:drl_compare}(b), the model-simplified H-DDPG has a higher variance in the reward performance, especially in the early stage of the learning. One possible reason is that the optimization in the model-simplified H-DDPG is based on an inexact time allocation $t$ generated by the inner-loop DDPG algorithm. Due to random exploration in the early stage, the time variable $t$ can change intensively and lead to very different optimization-driven target values. 

\begin{figure}[t]
    \centering
    \subfloat[Throughput performance]{\includegraphics[width=\singlesize\textwidth]{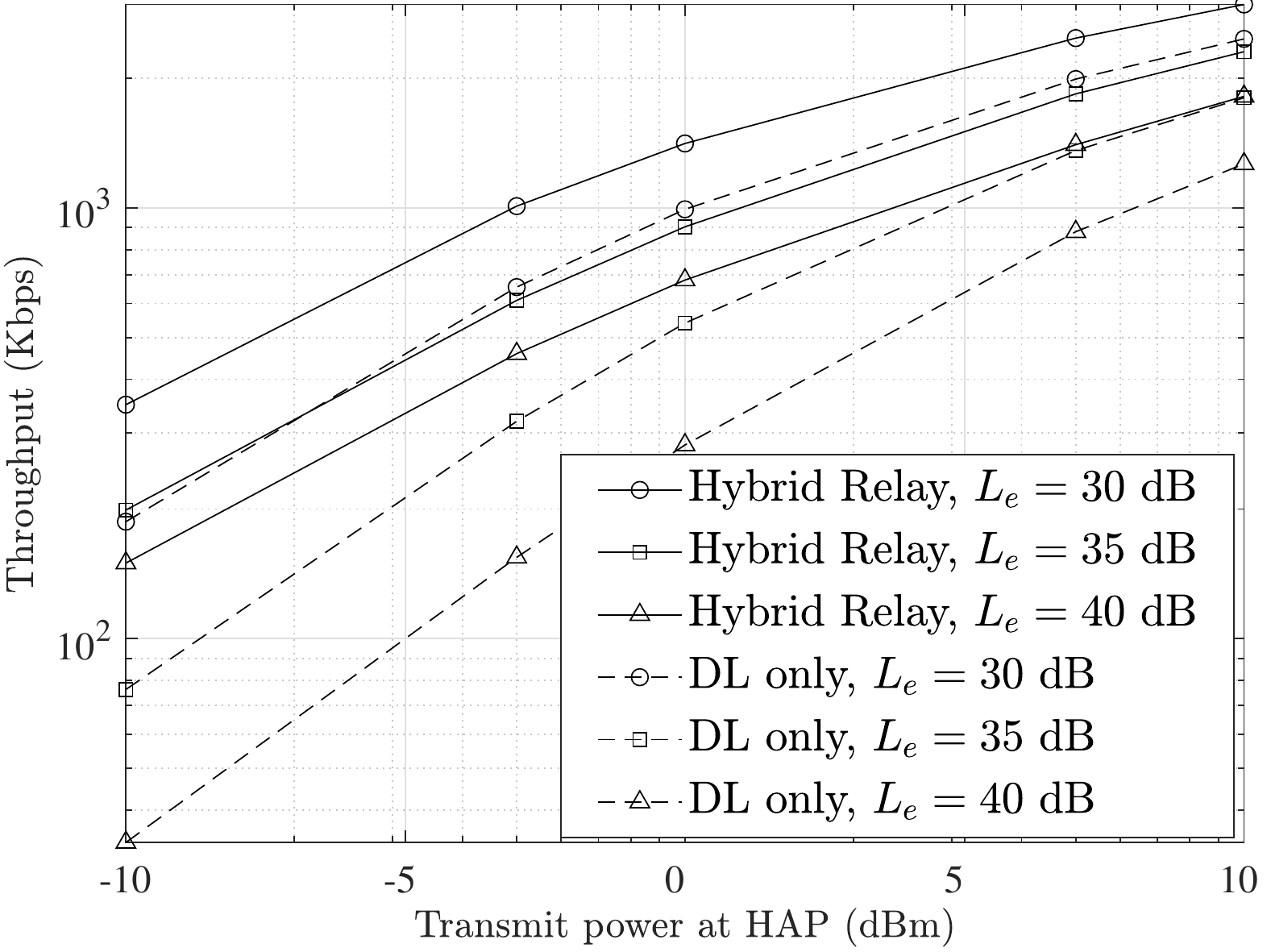}}\\
    \subfloat[Transmission time]{\includegraphics[width=\singlesize\textwidth]{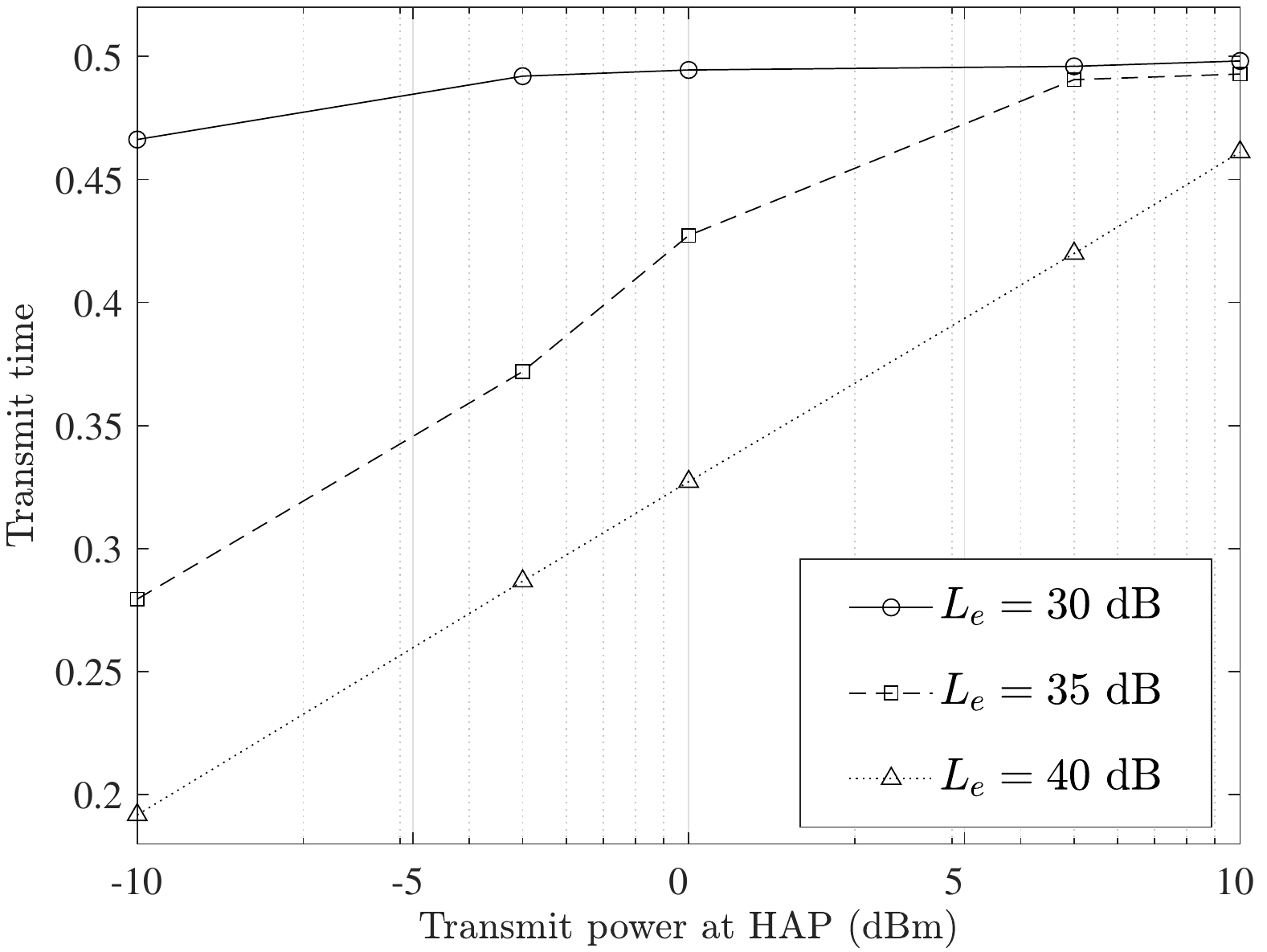}}
    \caption{Hybrid relays contribute more with a weak DL channel and a limited transmit power at the HAP.}\label{fig:dl-relay}
\end{figure}

\subsection{Throughput via Direct Link and Hybrid Relays}

In this part, we evaluate the optimal relay throughput with different transmit power at the HAP. For a fair comparison, we also show the optimal throughput when the direct link (DL) is available only. We consider an additional attenuation $L_e$ in the DL channel to indicate its channel quality. A higher value of $L_e$ implies a weak DL channel possibly due to the blockage of surrounding objects. A common observation in Fig.~\ref{fig:dl-relay}(a) is that the throughput increases with the HAP's transmit power. For a weaker DL channel, we observe that the relay-assisted throughput increases more significantly compared to that achievable via the DL only. In particular, given $L_e=40$ dB in the DL channel, the hybrid relay-assisted transmission provides $345\%$ higher throughput than that achievable via the DL channel when $p_t = -10$ dBm, which reduces to $43.48\%$ when the HAP's transmit power increases to $p_t=10$ dBm. The reason is that the hybrid relays will contribute little to the overall throughput when a strong DL channel can achieve high throughput performance. On the other hand, when the DL channel becomes worse off or the HAP's transmit power is limited, the proposed hybrid relay-assisted transmission can improve significantly the overall throughput and energy efficiency. Fig.~\ref{fig:dl-relay}(b) shows that the optimal transmission time also increases in the HAP's transmit power. This means that the HAP with a higher transmit power can schedule less channel time for wireless power transfer to the relays. With fixed transmit power at the HAP, we observe that the quality of DL channel also affects the relays' optimal energy harvesting time. In particular, with a weak DL channel, e.g.,~$L_e=40$ dB, the HAP spares more time to power the hybrid relays, as shown in Fig.~\ref{fig:dl-relay}(b). When the relays have more power, the relay transmission can contribute more to the overall throughput.

\begin{figure}[t]
    \centering
    \includegraphics[width=\singlesize\textwidth]{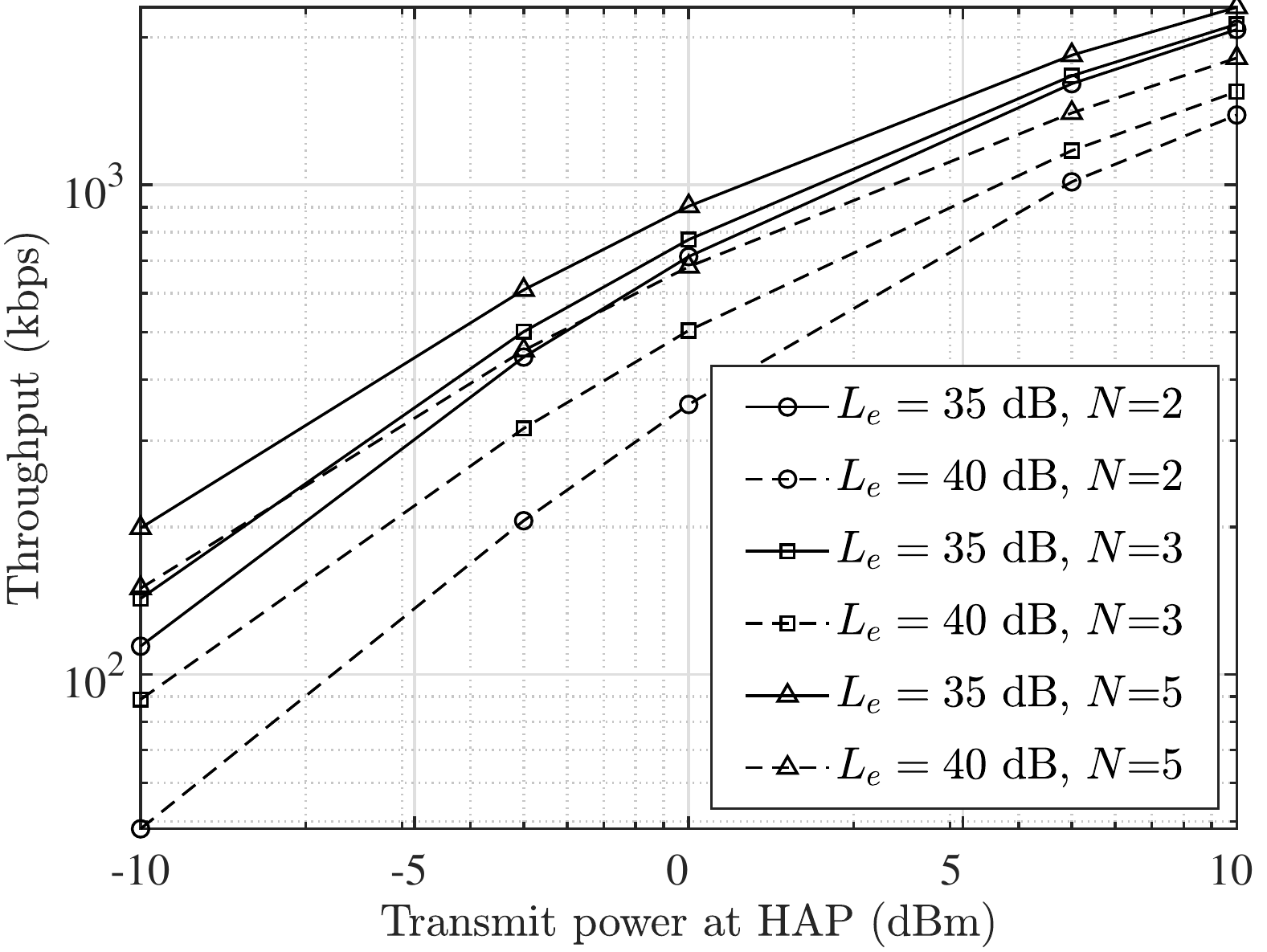}
    \caption{More relays can compensate the throughput loss in a weak DL channel.}\label{fig:relays_compensate}
\end{figure}

In Fig.~\ref{fig:relays_compensate}, we show the tradeoff between DL and relay transmissions as the number of hybrid relays increases. We can observe that more relays can compensate the throughput loss when the DL channel becomes worse off, especially for low transmit power at the HAP. For example, when the HAP's transmit power is $-10$ dBm, we can achieve a higher throughput in the case with $L_e = 40$ and $N=5$, comparing to the case with $L_e = 35$ and $N=2$ or $N=3$. This implies that the joint effort of more relays can contribute even better than a strong DL channel.

\subsection{Performance Comparison of Different Beamforming Schemes}

\begin{figure}[t]
    \centering
    \subfloat[Throughput comparison]{\includegraphics[width=\singlesize\textwidth]{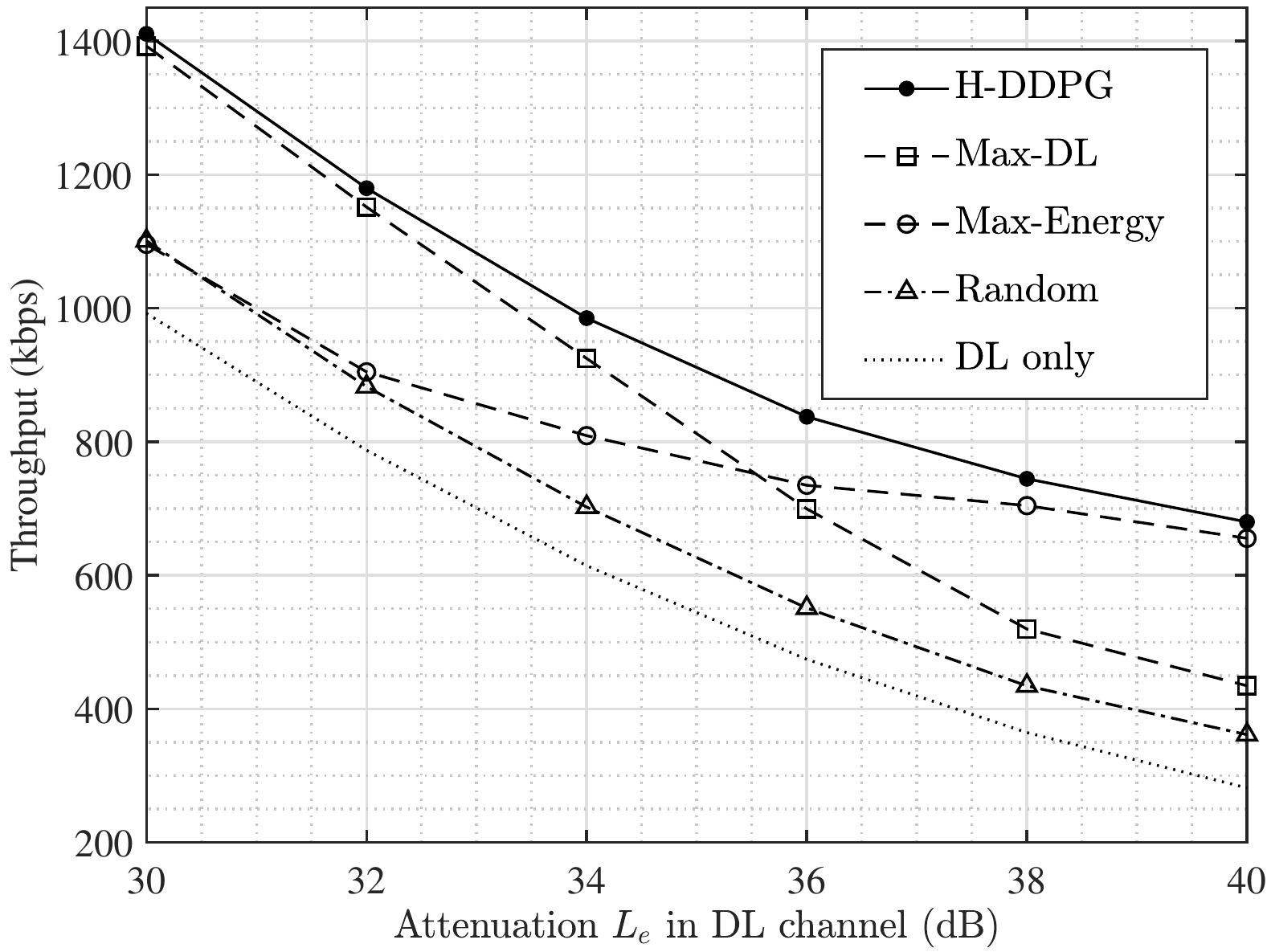}}\\
    \subfloat[Relay mode selection]{\includegraphics[width=\singlesize\textwidth]{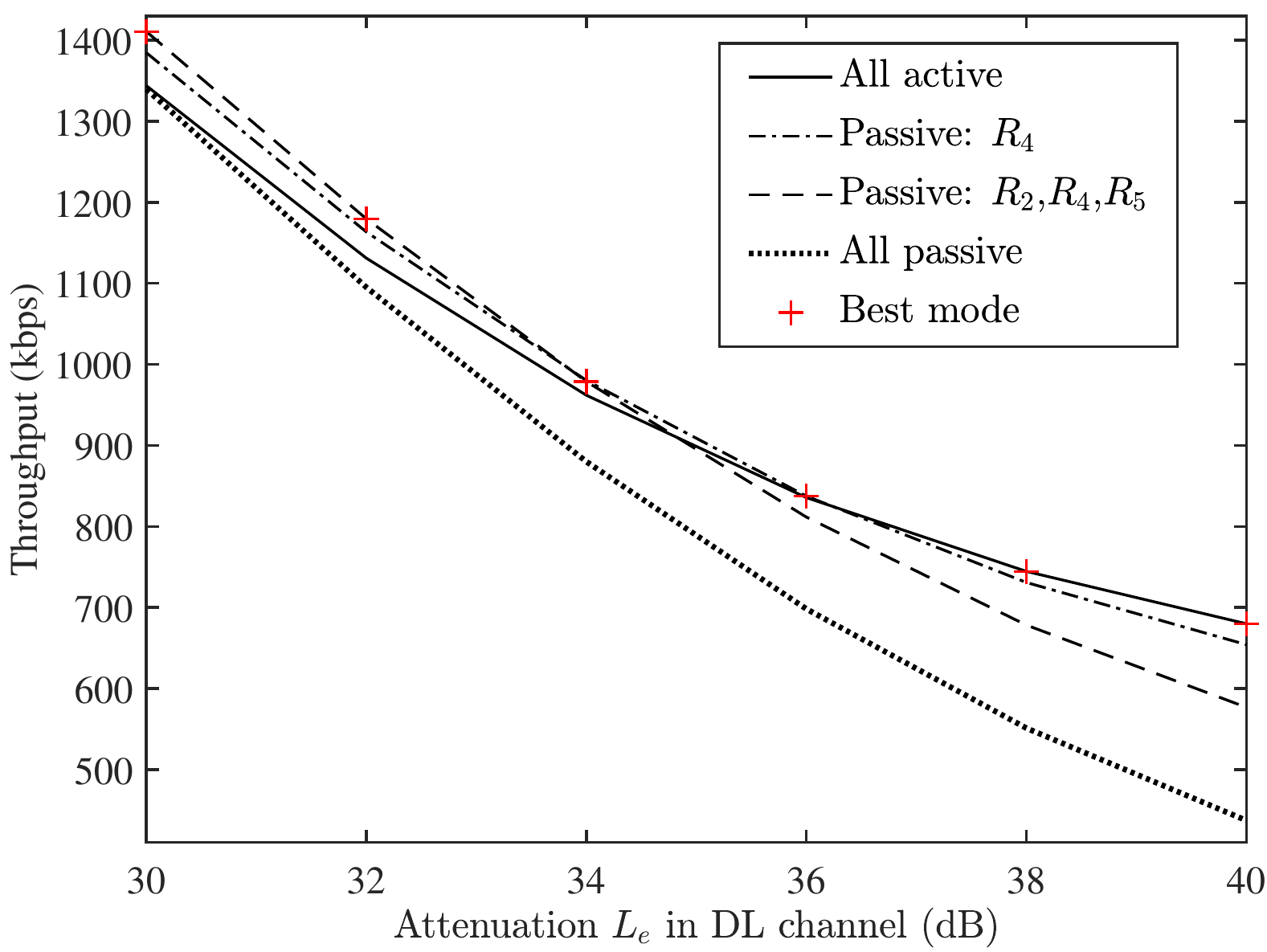}}
    \caption{The optimal throughput and the relays' mode selection change with the DL channel quality.}\label{fig:relay}
\end{figure}

In this part, we compare the optimal throughput achievable by the model-based H-DDPG algorithm with a few baselines as shown in Fig.~\ref{fig:relay}(a). To ensure a fair comparison under different channel conditions, we consider different DL channel qualities by varying the attenuation $L_e$ from $30$ dB to $40$ dB. Higher attenuation implies a weaker DL channel. In the comparison, we also implement three other beamforming algorithms: 1) The simplest random scheme allows the HAP to set its beamforming vector randomly, 2) Max-DL scheme aligns the HAP's beamformer to the DL channel and thus enhances direct information transmissions, 3) Max-Energy scheme adopts the optimal beamformer that maximizes the energy transfer to the relays, and also uses the same beamformer for information transmission in the second phase. As shown in Fig.~\ref{fig:relay}(a), the model-based H-DDPG achieves the optimal throughput performance. The Max-DL scheme achieves a similar throughput performance as that of optimal scheme when the DL channel is relatively good, i.e., the attenuation $L_e$ is small, while the performance gap becomes significant with the increase of $L_e$. On the other hand, the Max-Energy scheme shows a very different result. In particular, the performance gap between the Max-Energy scheme and the optimal scheme decreases with the increase of $L_e$. These observations imply that the model-based H-DDPG algorithm can adapt the HAP's beamforming according to different DL channel qualities. We also compare these beamforming schemes to the degenerated case with DL only. The results in Fig.~\ref{fig:relay}(a) show that the random beamforming scheme for the relay-assisted network even achieves a higher throughput than that achievable via the DL only. This verifies that the hybrid relays can provide the users' cooperation gain for information transmissions.

\begin{table}[t]
    \centering
    \caption{The hybrid relays' mode selection with $N=5$ relays}\label{tab:mode}
    \begin{tabular}{|c|c|c|c|c|}
    \hline
        DL channel quality & $p_t = -10$ dBm & $p_t = 0$ dBm & $p_t = 10$ dBm  \\ \hline
        $L_e = 30$ dB & $P_4$, $P_5$ & $P_2$, $P_4$, $P_5$ & $P_2$, $P_4$, $P_5$  \\ \hline
        $L_e = 35$ dB & A & $P_4$ & $P_2$, $P_4$, $P_5$  \\ \hline
        $L_e = 40$ dB & A & A & $P_4$, $P_5$  \\ \hline
    \end{tabular}
    \vspace{5mm}
    \caption{The hybrid relays' mode selection with $L_e=35$ dB}\label{tab:numR}
    \begin{tabular}{|c|c|c|c|}
    \hline
        Relay deployment & $p_t = -10$ dBm & $p_t = 0$ dBm & $p_t = 10$ dBm \\ \hline
        Relay-$1$ to Relay-$2$ & $P_2$ & $P_2$ & $P_2$ \\ \hline
        Relay-$1$ to Relay-$3$ & A & A & $P_1$, $P_2$ \\ \hline
        Relay-$1$ to Relay-$5$ & A & $P_4$ & $P_2$, $P_4$, $P_5$ \\ \hline
    \end{tabular}
\end{table}
\subsection{Hybrid Relays' Optimal Mode Selection}

Fig.~\ref{fig:relay}(b) shows the change of the relays' mode selection with different DL channel qualities. It is clear that the hybrid relay system can achieve the optimal throughput performance by allowing individual relays to optimize their operating modes. Moreover, we observe that more relays will switch to the active mode when the DL channel becomes worse off. This is further verified in Table~\ref{tab:mode}, where we enumerate the relays' optimal mode selection when the DL channel becomes worse off. We use 'A' in Table~\ref{tab:mode} to denote all active relays and $P_n$ denotes the passive mode of relay-$n$. For example, it becomes optimal for all relays to choose the active mode when $L_e=40$ dB. On the contrary, more relays will work in the passive mode when the DL channel becomes good. In this case, the passive relays prefer to assist DL transmissions and save more channel time for active RF transmissions. In Table~\ref{tab:numR}, we show the relays' mode selection with different HAP's transmit power and relays' deployment scheme. It is clear that the relays' optimal operating modes also change with the HAP's transmit power. An interesting finding is that the relays' deployment also affects their optimal radio mode. That is because the optimization-driven H-DDPG algorithm can capture the couplings between different relays and adapt their radio modes accordingly to maximize the overall throughput.

\section{Conclusions}\label{sec_con}

In this paper, we propose a novel optimization-driven hierarchical DRL approach to solve the throughput maximization problem involving both active and passive relays. This approach integrates DQN and model-based optimization methods into the conventional DDPG algorithm in a hierarchical learning framework. The model-based optimization can help derive a performance lower bound based on incomplete system information, which provides a better-informed target $Q$-value estimation for the online $Q$-network to follow and adapt in the learning process. We then deploy this learning framework in a multi-relay-assisted hybrid radio network to maximize its throughput performance. Simulation results verify that the proposed algorithm outperforms the model-free H-DDPG algorithm in terms of robustness and learning efficiency.

\bibliographystyle{IEEEtran}

\bibliography{ref}

\begin{thebibliography}{10}
\providecommand{\url}[1]{#1}
\csname url@samestyle\endcsname
\providecommand{\newblock}{\relax}
\providecommand{\bibinfo}[2]{#2}
\providecommand{\BIBentrySTDinterwordspacing}{\spaceskip=0pt\relax}
\providecommand{\BIBentryALTinterwordstretchfactor}{4}
\providecommand{\BIBentryALTinterwordspacing}{\spaceskip=\fontdimen2\font plus
\BIBentryALTinterwordstretchfactor\fontdimen3\font minus
  \fontdimen4\font\relax}
\providecommand{\BIBforeignlanguage}[2]{{%
\expandafter\ifx\csname l@#1\endcsname\relax
\typeout{** WARNING: IEEEtran.bst: No hyphenation pattern has been}%
\typeout{** loaded for the language `#1'. Using the pattern for}%
\typeout{** the default language instead.}%
\else
\language=\csname l@#1\endcsname
\fi
#2}}
\providecommand{\BIBdecl}{\relax}
\BIBdecl

\bibitem{lxmag18}
X.~Lu, D.~Niyato, H.~Jiang, D.~I. Kim, Y.~Xiao, and Z.~Han, ``Ambient
  backscatter assisted wireless powered communications,'' \emph{IEEE Wireless
  Commun.}, vol.~25, no.~2, pp. 170--177, Apr. 2018.

\bibitem{symbiotic}
R.~Long, H.~Guo, G.~Yang, Y.-C. Liang, and R.~Zhang, ``Symbiotic radio: A new
  communication paradigm for passive internet of things,'' \emph{IEEE Internet
  of Things Journal}, vol.~7, no.~2, pp. 1350--1363, Nov. 2019.

\bibitem{yangiot19}
G.~Yang, D.~Yuan, Y.-C. Liang, R.~Zhang, and V.~C.~M. Leung, ``Optimal resource
  allocation in full-duplex ambient backscatter communication networks for
  wireless-powered {IoT},'' \emph{IEEE Internet of Things Journal}, vol.~6,
  no.~2, pp. 2612--2625, Apr. 2019.

\bibitem{zhou19}
J.~{Guo}, X.~{Zhou}, and S.~{Durrani}, ``Wireless power transfer via {mmWave}
  power beacons with directional beamforming,'' \emph{IEEE Wireless Commun.
  Lett.}, vol.~8, no.~1, pp. 17--20, 2019.

\bibitem{hoang17cr}
D.~T. Hoang, D.~Niyato, P.~Wang, D.~I. Kim, and Z.~Han, ``Ambient backscatter:
  A new approach to improve network performance for {RF}-powered cognitive
  radio networks,'' \emph{IEEE Trans. Commun.}, vol.~65, no.~9, pp. 3659--3674,
  Sep. 2017.

\bibitem{tccn19}
S.~{Gong}, L.~{Gao}, J.~{Xu}, Y.~{Guo}, D.~T. {Hoang}, and D.~{Niyato},
  ``Exploiting backscatter-aided relay communications with hybrid access model
  in device-to-device networks,'' \emph{IEEE Trans. Cogn. Commun. Netw.},
  vol.~5, no.~4, pp. 835--848, 2019.

\bibitem{ieeenetwork}
S.~Gong, J.~Xu, D.~Niyato, X.~Huang, and Z.~Han, ``Backscatter-aided
  cooperative relay communications in wireless-powered hybrid radio networks,''
  \emph{IEEE Network}, vol.~33, pp. 234--241, Sep. 2019.

\bibitem{lyb-arelay}
B.~{Lyu} and D.~T. {Hoang}, ``Optimal time scheduling in relay assisted
  batteryless {IoT} networks,'' \emph{IEEE Wireless Commun. Lett.}, vol.~9,
  no.~5, pp. 706--710, 2020.

\bibitem{lyb-arelay2}
B.~{Lyu}, D.~T. {Hoang}, and Z.~{Yang}, ``Backscatter then forward: A relaying
  scheme for batteryless {IoT} networks,'' \emph{IEEE Wireless Commun. Lett.},
  vol.~9, no.~4, pp. 562--566, 2020.

\bibitem{twobds}
D.~{Li}, ``Two birds with one stone: Exploiting decode-and-forward relaying for
  opportunistic ambient backscattering,'' \emph{IEEE Trans. Commun.,}, vol.~68,
  no.~3, pp. 1405--1416, 2020.

\bibitem{R5}
B.~{Lyu}, D.~T. {Hoang}, and Z.~{Yang}, ``User cooperation in wireless-powered
  backscatter communication networks,'' \emph{IEEE Wireless Commun. Lett.},
  vol.~5, no.~3, pp. 2015--2024, Jun. 2018.

\bibitem{yang18}
G.~Yang, Q.~Zhang, and Y.~C. Liang, ``Cooperative ambient backscatter
  communications for green internet-of-things,'' \emph{IEEE Internet of Things
  Journal}, vol.~5, no.~2, pp. 1116--1130, Apr. 2018.

\bibitem{zhouber19}
X.~{Jia} and X.~{Zhou}, ``Decode-and-forward relaying using a backscatter
  device: Power allocation and {BER} analysis,'' in \emph{proc. IEEE GLOBECOM},
  2019, pp. 1--6.

\bibitem{zhou20}
------, ``Performance characterization of relaying using backscatter devices,''
  \emph{IEEE Open Journal of the Communication Society}, vol.~1, pp. 819--834,
  Jun. 2020.

\bibitem{luxiao19}
X.~{Lu}, D.~{Niyato}, H.~{Jiang}, E.~{Hossain}, and P.~{Wang}, ``Ambient
  backscatter-assisted wireless-powered relaying,'' \emph{IEEE Trans. Green
  Commun. Network.}, vol.~3, no.~4, pp. 1087--1105, 2019.

\bibitem{gc19xie}
Y.~{Xie}, Z.~{Xu}, S.~{Gong}, J.~{Xu}, D.~T. {Hoang}, and D.~{Niyato},
  ``Backscatter-assisted hybrid relaying strategy for wireless powered {IoT}
  communications,'' in \emph{proc. IEEE GLOBECOM}, 2019, pp. 1--6.

\bibitem{iot20}
S.~{Gong}, Y.~{Zou}, D.~T. {Hoang}, J.~{Xu}, W.~{Cheng}, and D.~{Niyato},
  ``Capitalizing backscatter-aided hybrid relay communications with wireless
  energy harvesting,'' \emph{IEEE Internet of Things Journal}, pp. 1--1, 2020.

\bibitem{dsa-rl}
N.~{Van Huynh}, D.~T. {Hoang}, D.~N. {Nguyen}, E.~{Dutkiewicz}, D.~{Niyato},
  and P.~{Wang}, ``Optimal and low-complexity dynamic spectrum access for
  {RF}-powered ambient backscatter system with online reinforcement learning,''
  \emph{IEEE Trans. Commun.}, vol.~67, no.~8, pp. 5736--5752, 2019.

\bibitem{jam-hoang}
N.~{Van Huynh}, D.~N. {Nguyen}, D.~T. {Hoang}, and E.~{Dutkiewicz}, ``"{J}am me
  if you can:" defeating jammer with deep dueling neural network architecture
  and ambient backscattering augmented communications,'' \emph{IEEE J. Sel.
  Area. Commun.}, vol.~37, no.~11, pp. 2603--2620, 2019.

\bibitem{rl-jam}
N.~{Van Huynh}, D.~N. {Nguyen}, D.~{Thai Hoang}, E.~{Dutkiewicz}, and
  M.~{Mueck}, ``Ambient backscatter: A novel method to defend jamming attacks
  for wireless networks,'' \emph{IEEE Wireless Commun. Lett.}, vol.~9, no.~2,
  pp. 175--178, 2020.

\bibitem{drl-time}
T.~T. {Anh}, N.~C. {Luong}, D.~{Niyato}, Y.~{Liang}, and D.~I. {Kim}, ``Deep
  reinforcement learning for time scheduling in {RF}-powered backscatter
  cognitive radio networks,'' in \emph{proc. IEEE WCNC}, 2019, pp. 1--7.

\bibitem{q-power}
F.~{Jameel}, W.~U. {Khan}, S.~T. {Shah}, and T.~{Ristaniemi}, ``Towards
  intelligent {IoT} networks: Reinforcement learning for reliable backscatter
  communications,'' in \emph{proc. IEEE GLOBECOM Workshops}, 2019, pp. 1--6.

\bibitem{iccc-ddpg}
Y.~{Xie}, Z.~{Xu}, Y.~{Zhong}, J.~{Xu}, S.~{Gong}, and Y.~{Wang},
  ``Backscatter-assisted computation offloading for energy harvesting {IoT}
  devices via policy-based deep reinforcement learning,'' in \emph{proc.
  IEEE/CIC ICCC Workshops}, 2019, pp. 65--70.

\bibitem{ddpg}
T.~P. Lillicrap, J.~J. Hunt, A.~Pritzel, N.~Heess, T.~Erez, Y.~Tassa,
  D.~Silver, and D.~Wierstra, ``Continuous control with deep reinforcement
  learning,'' in \emph{proc. Int. Conf. Learning Representations (ICLR)}, San
  Juan, Puerto Rico, May 2016.

\bibitem{luong18}
N.~C. Luong, D.~T. Hoang, S.~Gong, D.~Niyato, P.~Wang, Y.~Liang, and D.~I. Kim,
  ``Applications of deep reinforcement learning in communications and
  networking: a survey,'' \emph{IEEE Commun. Surv. Tut.}, vol.~21, no.~4, pp.
  3133--3174, May 2019.

\bibitem{hddpg}
Y.~Zou, Y.~Xie, C.~Zhang, L.~Guo, S.~Gong, H.~Dinh, and D.~Niyato,
  ``Optimization-driven hierarchical deep reinforcement learning for hybrid
  relaying communications,'' in \emph{proc. IEEE WCNC}, May 2020, pp. 1--6.

\bibitem{gsmwcnc19}
X.~Luo, J.~Xu, Y.~Zou, S.~Gong, L.~Gao, and D.~Niyato, ``Collaborative relay
  beamforming with direct links in wireless powered communications,'' in
  \emph{proc. IEEE WCNC}, Apr. 2019, pp. 1--6.

\bibitem{monot13}
Y.~J.~A. Zhang, L.~Qian, and J.~Huang, ``Monotonic optimization in
  communication and networking systems,'' \emph{Foundations and
  Trends\textregistered~in Networking}, vol.~7, no.~1, pp. 1--75, 2013.

\bibitem{yztcom}
J.~{Xu}, Y.~{Zou}, S.~{Gong}, L.~{Gao}, D.~{Niyato}, and W.~{Cheng}, ``Robust
  transmissions in wireless-powered multi-relay networks with chance
  interference constraints,'' \emph{IEEE Trans. Commun.}, vol.~67, no.~2, pp.
  973--987, Feb. 2019.

\bibitem{sutton}
R.~S. Sutton and A.~G. Barto, \emph{Reinforcement Learning: An
  Introduction}.\hskip 1em plus 0.5em minus 0.4em\relax The MIT Press, 2018.

\bibitem{rankluo}
Z.-Q. Luo, W.-K. Ma, A.-C. So, Y.~Ye, and S.~Zhang, ``Semidefinite relaxation
  of quadratic optimization problems,'' \emph{IEEE Signal Process. Mag.},
  vol.~27, no.~3, pp. 20--34, May 2010.

\end{thebibliography}

\end{document}